\newif\ifdraft{}
\newif\iffull{}
\newif\ifsubmission{}
\newif\ifieee{}
\newif\ifacm{}

\acmfalse{}
\drafttrue{}
\ieeefalse{}
\fullfalse{}
\iffull{}
  \submissionfalse{}
\else
  \submissiontrue{}
  \ieeetrue{}
\fi

\ifsubmission{}
  \documentclass[conference]{IEEEtran}
\else
  \documentclass[11pt, english]{article}
  \usepackage[margin=1in]{geometry}
\fi
\usepackage[T1]{fontenc}
\usepackage{paralist}
\ifieee{}
\usepackage[cmex10]{amsmath}
\usepackage{amssymb,amsfonts,amsthm}
\else
\usepackage{amsmath,amsfonts,amsthm}
\fi
\ifieee{}
\fi
\usepackage{algorithm}
\usepackage[noend]{algpseudocode}
\usepackage{graphicx}
\usepackage{booktabs}
\usepackage{soul}

\usepackage{tabularx}
\usepackage{varwidth}
\usepackage{array}
\usepackage{adjustbox}
\usepackage{subcaption}
\usepackage{mwe}
\usepackage{textcomp}
\usepackage{comment}
\usepackage{multirow}
\usepackage{multicol}
\usepackage{mathtools}
\usepackage{dirtytalk}
\usepackage{csquotes}
\usepackage{xspace}
\usepackage{cite}
\usepackage{paralist}

\usepackage{thmtools}
\usepackage[english]{babel}
\usepackage{float}
\usepackage{pifont}
\usepackage{capt-of}
\iffull{}
  \usepackage[hyphens]{url}
\else
  \usepackage{url}
\fi
\usepackage[inline]{enumitem}
\usepackage[hidelinks]{hyperref}
\hypersetup{
	colorlinks=true,
	linkcolor=blue,
	urlcolor=blue,
	citecolor=blue,
}
\usepackage[capitalize]{cleveref}
\iffull{}
\usepackage[breaklinks]{hyperref}
\fi
\iffull{}
  \usepackage{balance}
\fi
\usepackage{array}
\newtheorem{theorem}{Theorem}[section]

\newtheorem{lemma}[theorem]{Lemma}
\newtheorem{definition}{Definition}[section]

\usepackage{tikz}
\usetikzlibrary{backgrounds}
\usetikzlibrary{shapes,shapes.geometric,shapes.symbols,arrows,positioning}
\tikzstyle{block} = [rectangle, draw, text width=1cm, 
                     text centered, rounded corners, 
                     minimum height=1cm, node distance=1cm]
\tikzstyle{process} = [circle, draw, text width=1cm, 
                     text centered, rounded corners, 
                     node distance=1cm]
\tikzstyle{line} = [draw, -latex']
\tikzstyle{bubble} = [cloud, draw, text width=1cm, 
                     text centered, rounded corners, 
                     minimum height=1cm, node distance=2cm]

\iffull{}
\usepackage{changepage}
\fi

\definecolor{myorange}{HTML}{d95319}
\makeatletter
\algnewcommand{\LineComment}[1]{\Statex{}\hskip\ALG@thistlm{}{\color{gray}\textrm{// #1}}}
\makeatother
\algnewcommand{\SectionComment}[2]{\Statex{}{\color{#2}\(\triangleright\) \textrm{#1}}}

\algnewcommand{\InlineComment}[1]{{\hspace{0.5em}\color{gray}\textrm{// #1}}}
\algnewcommand{\InlineCommentText}[1]{{\hspace{0.5em}\color{gray}{// #1}}}
\algnewcommand{\Phase}[1]{\SectionComment{#1}{myorange}}
\MakeRobust{\Call}
\DeclareCaptionFormat{algor}{%
	\hrulefill\par\offinterlineskip\vskip1pt%
	\textbf{#1#2}#3\offinterlineskip\hrulefill}
\DeclareCaptionStyle{algori}{singlelinecheck=off,format=algor,labelsep=space}
\algnewcommand\algorithmicswitch{\textbf{switch}}
\algnewcommand\algorithmiccase{\textbf{case}}
\algnewcommand\algorithmicassert{\texttt{assert}}
\algnewcommand\Assert[1]{\State{}\algorithmicassert(#1)}%
\algdef{SE}[SWITCH]{Switch}{EndSwitch}[1]{\algorithmicswitch\ #1\ \algorithmicdo}{\algorithmicend\ \algorithmicswitch}%
\algdef{SE}[CASE]{Case}{EndCase}[1]{\algorithmiccase\ #1}{\algorithmicend\ \algorithmiccase}%
\algtext*{EndSwitch}%
\algtext*{EndCase}
\algnewcommand\algorithmiccontinue{\textbf{continue}}
\algnewcommand\algorithmicbreak{\textbf{break}}
\algnewcommand\Continue{\algorithmiccontinue}
\algnewcommand\Break{\algorithmicbreak}
\algnewcommand{\IIf}[1]{\State\algorithmicif\ #1\ \algorithmicthen}
\algnewcommand{\EndIIf}{\unskip\ \algorithmicend\ \algorithmicif}
\algblock{As}{EndAs}
\algblock{On}{EndOn}
\algnewcommand\algorithmicas{\textbf{as}}
\algrenewtext{As}[1]{\algorithmicas\ #1}
\algtext*{EndAs}
\algnewcommand\algorithmicon{\textbf{on}}
\algrenewtext{On}[1]{\algorithmicon\ #1}
\algtext*{EndOn}

\usepackage{multirow}

\begin{document}

\def\titletext{
    Delphi: Efficient Asynchronous Approximate Agreement for Distributed Oracles}

\title{\titletext}
\author{
	\IEEEauthorblockN{Akhil Bandarupalli\IEEEauthorrefmark{1}, Adithya Bhat\IEEEauthorrefmark{2}, Saurabh Bagchi\IEEEauthorrefmark{1}, Aniket Kate\IEEEauthorrefmark{1}\IEEEauthorrefmark{3}, Chen-Da Liu-Zhang\IEEEauthorrefmark{4}\IEEEauthorrefmark{5}, Michael K. Reiter\IEEEauthorrefmark{6}\IEEEauthorrefmark{7}}
	\IEEEauthorblockA{\IEEEauthorrefmark{1}Purdue University \emph{\{abandaru, sbagchi, aniket\}@purdue.edu}}
	\IEEEauthorblockA{\IEEEauthorrefmark{2}Visa Research \emph{haxolotl.research@gmail.com}}
	\IEEEauthorblockA{\IEEEauthorrefmark{4}Lucerne University of Applied Sciences and Arts \emph{chen-da.liuzhang@hslu.ch}}
	\IEEEauthorblockA{\IEEEauthorrefmark{6}Duke University \emph{michael.reiter@duke.edu}}
	\IEEEauthorblockA{\IEEEauthorrefmark{3}Supra Research, \IEEEauthorrefmark{5}Web3 Foundation,
    \IEEEauthorrefmark{7}Chainlink Labs}
}


\maketitle
\newcommand{\name}{{\sc{Delphi}}\xspace}


\newcommand*{\chop}[1]{}
\def\honestrange{\ensuremath{\delta}\xspace}
\def\honestrangex{\ensuremath{\delta_{x}}\xspace}
\def\honestrangey{\ensuremath{\delta_{y}}\xspace}
\def\maxrange{\ensuremath{\Delta}\xspace}
\newcommand{\distanceerror}[1]{\ensuremath{\vec{d_{#1}}}\xspace}
\newcommand{\distanceerrorx}[1]{\ensuremath{\vec{d_{#1}}.x}\xspace}
\newcommand{\distanceerrory}[1]{\ensuremath{\vec{d_{#1}}.y}\xspace}
\newcommand{\iou}[1]{\ensuremath{I_{#1}}\xspace}
\newcommand{\diagonalbb}{\ensuremath{l_{\mathsf{diag}}}\xspace}
\newcommand{\propconst}{\ensuremath{}\xspace}

\newcommand{\adv}{\ensuremath{\mathcal{A}}\xspace}
\def\negl{\ensuremath{\mathsf{negl}}\xspace}
\newcommand{\secparam}{\ensuremath{\kappa}\xspace}
\def\netz{\ensuremath{\mathcal{Z}}\xspace}
\newcommand{\statparam}{\ensuremath{\lambda}\xspace}
\newcommand{\hashfunc}{\ensuremath{H}\xspace}

\newcommand{\oracleoutput}[1]{\ensuremath{o_{#1}}\xspace}
\newcommand{\honestnodes}{\ensuremath{\mathcal{P}}\xspace}
\newcommand{\honestset}{\ensuremath{V_{\mathsf{h}}}\xspace}
\newcommand{\oracleinput}[1]{\ensuremath{v_{#1}}\xspace}
\newcommand{\bitinput}[1]{\ensuremath{b_{#1}}\xspace}
\newcommand{\checkpoint}{\ensuremath{c}\xspace}
\newcommand{\closestcc}{\ensuremath{z}\xspace}
\newcommand{\midpoint}{\ensuremath{\mu}\xspace}
\newcommand{\weight}{\ensuremath{\omega}\xspace}
\newcommand{\levelmax}{\ensuremath{l_{\mathsf{M}}}\xspace}
\newcommand{\level}[1]{\ensuremath{l_{#1}}\xspace}
\newcommand{\round}[1]{\ensuremath{r_{#1}}\xspace}
\newcommand{\roundmax}{\ensuremath{r_{\mathsf{M}}}\xspace}
\newcommand{\anchor}{\ensuremath{a}\xspace}
\newcommand{\levelstart}{\ensuremath{s}\xspace}
\newcommand{\levelend}{\ensuremath{e}\xspace}
\newcommand{\levelset}{\ensuremath{\mathbb{L}}\xspace}
\newcommand{\levelstate}[1]{\ensuremath{\mathcal{L}[#1]}\xspace}
\newcommand{\intmidpoint}[1]{\ensuremath{\midpoint_{#1}}\xspace}
\newcommand{\interval}[3]{\ensuremath{\mathcal{C}_{#1}^{#3}}\xspace}
\newcommand{\intervalmidpoint}[3]{\ensuremath{\midpoint_{#1}^{#3}}\xspace}
\newcommand{\binaastate}[3]{\ensuremath{\mathcal{B}_{#1}^{#3}}\xspace}
\newcommand{\binaastateval}[3]{{\weight_{#1}^{#3}}\xspace}
\newcommand{\binaastatevalmsg}[5]{{M\langle\binaastateval{#1}{#2}{#3},#4,#5\rangle}\xspace}
\newcommand{\binaastateweight}[3]{\ensuremath{w_{#1}^{#3}}\xspace}
\newcommand{\binaastatemsg}[3]{\binaastate{#1}{#2}{#3}.{\sc{Msg}}()\xspace} 
\newcommand{\levelmsgs}[1]{\ensuremath{\mathcal{M}[#1]}\xspace}
\newcommand{\levelmsgweight}[3]{\ensuremath{\levelmsgs{#3}[\interval{#1}{#2}{#3}].w}}
\newcommand{\levelmsg}{\ensuremath{\mathcal{M}}\xspace}
\newcommand{\levelencap}{\ensuremath{\mathbb{M}}\xspace}
\newcommand{\levelval}[1]{\ensuremath{V_{#1}}\xspace}
\newcommand{\levelweight}[1]{\ensuremath{w_{#1}}\xspace}
\newcommand{\lpweight}[1]{\ensuremath{w'_{#1}}\xspace}
\newcommand{\critlevel}{\ensuremath{\phi}\xspace}
\newcommand{\adddiff}[1]{\ensuremath{c_{#1}}\xspace}
\newcommand{\intweight}[1]{\ensuremath{\weight_{#1}}\xspace}
\newcommand{\multlevel}[4]{\ensuremath{\mathbb{V}\langle#1,#2,#3,#4\rangle}}
\newcommand{\variable}{\ensuremath{\tau}\xspace}

\newcommand{\echolist}[2]{\ensuremath{E1_{#1}^{#2}}\xspace}
\newcommand{\echotwolist}[2]{\ensuremath{E2_{#1}^{#2}}\xspace}

\newcommand{\msg}[1]{\ensuremath{\langle#1\rangle}\xspace}
\newcommand{\startmsg}{\ensuremath{\langle} {\sc{Start}}\ensuremath{,i\rangle}\xspace} 
\newcommand{\newround}[1]{\ensuremath{\langle}{\sc{BinAAMsg}}\ensuremath{,#1\rangle}\xspace}
\newcommand{\levelmsgwrapper}[2]{{\sc{LevelMsg}}\ensuremath{\langle}{#1,#2}\ensuremath{\rangle}\xspace}

\newcommand{\binaainit}[1]{{\sc{BinAA}}\ensuremath{(#1)}\xspace}
\newcommand{\binaaupdate}[4]{\ensuremath{#1}.{\sc{Update}}\ensuremath{(#2,#3,#4)}\xspace}
\newcommand{\binaaterm}[2]{\ensuremath{#1}.{\sc{EndRound}}\ensuremath{(#2)}\xspace}

\newcommand{\coalesce}{{\sc{Coalesce}}\xspace}
\newcommand{\decompose}{{\sc{Split}}\xspace}
\newcommand{\expand}{{\sc{Expand}}\xspace}

\def\rand{\ensuremath{R}\xspace}
\newcommand{\randj}[1]{\ensuremath{\rand_{#1}}\xspace}
\newcommand{\beacon}[1]{\ensuremath{\langle#1,b_{#1}\rangle}\xspace}
\newcommand{\beaconelem}[1]{\ensuremath{b_{#1}}\xspace}
\newcommand{\beaconprot}{\ensuremath{\mathcal{B}}\xspace}
\newcommand{\beacprepprot}[2]{\beaconprot.{\sc{Prep}}(#1,#2)\xspace} 
\newcommand{\beacopenprot}[1]{\beaconprot.{\sc{Open}}(#1)\xspace} 
\newcommand{\beaconqueue}{\ensuremath{\mathsf{B}}\xspace}

\newcommand{\sizeof}[1]{\ensuremath{\left\vert #1 \right\vert}\xspace}
\newcommand{\set}[1]{\ensuremath{\mathopen{}\{#1\}\mathclose{}}\xspace}
\newcommand{\paren}[1]{\ensuremath{\mathopen{}\left(#1 \right)\mathclose{}}}
\newcommand{\p}[1]{\ensuremath{(#1)}\xspace}
\newcommand{\tup}[1]{\ensuremath{\left\langle #1 \right\rangle}\xspace}
\newcommand{\hypergeom}{\ensuremath{\mathcal{H}}\xspace}

\newcommand{\unbox}{\floatstyle{plain}\restylefloat{figure}}
\newcommand{\rebox}{\floatstyle{boxed}\restylefloat{figure}}
\def\scalefigure{1.0}

\iffull{}
\renewcommand{\paragraph}[1]{\medskip\noindent\textbf{#1.}\hspace{1pt}}
\else
\renewcommand{\paragraph}[1]{\noindent\textbf{#1.}}
\fi

\def\etal{et al.\xspace}

\newcommand{\rbcprot}{\ensuremath{RBC}\xspace}
\newcommand{\rbroadcast}[1]{\rbcprot.{\sc{Broadcast}}(#1)\xspace} 
\newcommand{\rbcinit}[1]{\rbcprot.{\sc{Init}}(#1)\xspace} 
\newcommand{\rdeliver}[1]{{\rbcprot.{\sc{Deliver}}(#1)}\xspace} 

\newcommand{\dealer}{\ensuremath{n_d}\xspace}
\newcommand{\nodeset}{\ensuremath{\mathcal{P}}\xspace}

\newcommand{\realnums}{\mathbb{R}\xspace}
\newcommand{\natnums}{\mathbb{N}\xspace}
\newcommand{\hmessage}[1]{\ensuremath{M_{#1}}\xspace}
\newcommand{\domain}{\ensuremath{\mathcal{D}}\xspace}
\newcommand{\beacondomain}{\ensuremath{\mathcal{D}'}\xspace}
\newcommand{\polyeval}{\ensuremath{\mathcal{E}}\xspace}
\newcommand{\bottom}{\ensuremath{\perp}\xspace}
\newcommand{\aaprecision}{\ensuremath{\epsilon'}\xspace}
\newcommand{\aatermrounds}{\ensuremath{\round_{t}}\xspace}
\newcommand{\roundvalue}[2]{\ensuremath{V_{#1}^{#2}}\xspace}
\newcommand{\roundrange}[2]{\ensuremath{\Delta_{#1}^{#2}}\xspace}
\newcommand{\highthreshold}{\ensuremath{n-t}\xspace}
\newcommand{\noderange}[2]{\ensuremath{\{#1,\hdots,#2\}}\xspace}
\newcommand{\faults}{\ensuremath{t}\xspace}
\newcommand{\minthreshold}{\ensuremath{t+1}\xspace}

\newcommand{\pk}[1]{\ensuremath{pk_{#1}}\xspace}

\newcommand{\consensusoutput}{\ensuremath{z}\xspace}
\newcommand{\consensusvariable}{\ensuremath{\mathcal{Z}}\xspace}
\newcommand{\Protocol}{\ensuremath{\mathbb{P}}\xspace}
\newcommand{\estimaterounds}[2]{\ensuremath{EstimateRounds_{#1}(#2)}\xspace}
\newcommand{\reduce}[2]{\ensuremath{Reduce(#1,#2)}\xspace}
\newcommand{\sort}[1]{\ensuremath{Sort(#1)}\xspace}
\newcommand{\trim}[2]{\ensuremath{Trim(#1,#2)}\xspace}
\newcommand{\minimum}[1]{\ensuremath{Min(#1)}\xspace}
\newcommand{\maximum}[1]{\ensuremath{Max(#1)}\xspace}
\newcommand{\median}[1]{\ensuremath{Median(#1)}\xspace}
\newcommand{\setop}[1]{\ensuremath{\{#1\}}\xspace}
\newcommand{\sizeofset}[1]{\ensuremath{|#1|}\xspace}
\newcommand{\bigO}{\ensuremath{\mathcal{O}}\xspace}
\newcommand{\smallO}[1]{\ensuremath{o}(#1)\xspace}
\newcommand{\overshoot}[1]{\ensuremath{\rho_{#1}}\xspace}
\newcommand{\exponent}{\ensuremath{2}\xspace}
\newcommand{\bvoutputset}[1]{\ensuremath{B_{#1}}\xspace}

\newcommand{\bigo}[1]{\ensuremath{\mathcal{O}(#1)}\xspace}
\newcommand\tab[1][1cm]{\hspace*{#1}}

\newcommand{\protocol}{{\sc{\name}}\xspace}
\newcommand{\smrprotocol}{{\sc{PQ-Tusk}}\xspace}
\newcommand{\beaconprotocol}{{\sc{HashPipe}}\xspace}
\newcommand{\bundledaa}{{\sc{BunAA}}\xspace}
\newcommand{\binaryaa}{{\sc{BinAA}}\xspace}
\newcommand{\dagrider}{{\sc{DAG-Rider}}\xspace}
\newcommand{\fin}{{\sc{FIN}}\xspace}
\newcommand{\tusk}{{\sc{Tusk}}\xspace}
\newcommand{\picnic}{{\sc{Picnic}}\xspace}
\newcommand{\bawss}{{\sc{BAwVSS}}\xspace}
\newcommand{\awvss}{{\sc{AwVSS}}\xspace}
\newcommand{\avss}{{\sc{AVSS}}\xspace}
\newcommand{\pavss}{{\sc{PAVSS}}\xspace}
\newcommand{\anytrust}{AnyTrust\xspace}

\newcommand{\gatherprot}{{\sc{Gather}}\xspace}
\newcommand{\gatherdef}{{Gather}\xspace}
\newcommand{\gatherstart}{\gatherprot.{\sc{Start}}\xspace}
\newcommand{\gatherterm}[1]{\gatherprot.{\sc{Term}(#1)}\xspace} 
\newcommand{\gathercoreset}{\ensuremath{G}\xspace}
\newcommand{\gatherset}[1]{\ensuremath{G_{#1}}\xspace}
\newcommand{\totalityprot}[1]{\ensuremath{\mathcal{T}_{#1}}\xspace}

\newcommand{\committee}{\ensuremath{\mathsf{N}}\xspace}
\newcommand{\batchsize}{\ensuremath{\beta}\xspace}
\newcommand{\frequency}{\ensuremath{\phi}\xspace}
\newcommand{\commelectproc}{{\sc{ComElect}}\xspace}

\newcommand{\eaggrprot}{\ensuremath{\mathcal{E}}\xspace}
\newcommand{\approxconsensus}{Approximate Agreement\xspace}
\newcommand{\aastart}[1]{\ensuremath{\eaggrprot}.{\sc{Start}}(#1)\xspace} 
\newcommand{\aaterm}[1]{\ensuremath{\eaggrprot}.{\sc{Term}}(#1)\xspace} 
\newcommand{\aastartround}[2]{\ensuremath{\eaggrprot_{#1}}.{\sc{StartRound}(#2)}\xspace} 
\newcommand{\aaendround}[2]{\ensuremath{\eaggrprot_{#1}}.{\sc{EndRound}(#2)}\xspace} 
\newcommand{\aainitvalue}[1]{\ensuremath{m_{#1}}\xspace}
\newcommand{\aadecvalue}[1]{\ensuremath{o_{#1}}\xspace}
\newcommand{\aainitvalueset}{\ensuremath{\mathcal{M}}\xspace}
\newcommand{\appxconweights}[2]{\ensuremath{w_{#1,#2}}\xspace}
\newcommand{\queueappxcon}{\ensuremath{\mathsf{E}}\xspace}
\newcommand{\appxconoutputset}[1]{\ensuremath{\mathit{O_{#1}}}\xspace}
\newcommand{\checkpointindex}{\ensuremath{k}\xspace}
\newcommand{\checkpointindices}{\ensuremath{K}\xspace}
\newcommand{\checkpointcardi}{\ensuremath{\sizeof{U}_{\mathsf{\max}}}\xspace}

\newcommand{\locationcalc}[1]{\ensuremath{\vec{L_{\mathsf{#1}}}}\xspace}
\newcommand{\locationcalcx}[1]{\ensuremath{\vec{L_{\mathsf{#1}}}.x}\xspace}
\newcommand{\locationcalcy}[1]{\ensuremath{\vec{L_{\mathsf{#1}}}.y}\xspace}
\newcommand{\droneheight}{\ensuremath{h}\xspace}

\newcommand{\bracha}{{\sc{Bracha and Toueg}}\xspace}
\newcommand{\cachinrbc}{Cachin and Tessaro\xspace}
\newcommand{\dasrbc}{Das et al.\xspace}
\newcommand{\fktprotocol}{Freitas et al.\xspace}
\newcommand{\abrahamapprox}{Abraham et al.\xspace}

\newcommand{\expectation}[1]{\ensuremath{E[#1]}\xspace}

\newcommand{\sign}[2]{\ensuremath{\mathit{Sign}_{#1}(#2)}\xspace}
\newcommand{\verify}[2]{\ensuremath{\mathit{Verf}_{#1}{#2}}\xspace}

\newcommand{\cmark}{\ding{51}}%
\newcommand{\xmark}{\ding{55}}%

\newcommand{\sendall}[1]{\textbf{SendAll}{ #1}\xspace}
\newcommand{\receive}[2]{\ensuremath{\mathit{Receive}_{#1}{(#2)}}\xspace}
\newcommand{\range}[1]{\ensuremath{\mathit{\delta(#1)}}\xspace}
 \newcommand{\defeq}[0]{\ensuremath{{\;\vcentcolon=\;}}\xspace}

\newcommand{\upon}{\ensuremath{\mathit{\textbf{upon}}}\xspace}

\newcommand{\ceil}[1]{\ensuremath{\left\lceil{#1}\right\rceil}}
\newcommand{\floor}[1]{\ensuremath{\left\lfloor{#1}\right\rfloor}}

\makeatletter
\algnewcommand{\HorizontalLine}{\Statex\hskip\ALG@thistlm{}\hrulefill}
\makeatother

\algdef{SE}[SUBALG]{Indent}{EndIndent}{}{\algorithmicend\ }%
\algtext*{Indent}
\algtext*{EndIndent}
\newcommand{\algrule}[1][.2pt]{\par\vskip.5\baselineskip\hrule height #1\par\vskip.5\baselineskip}

\newcommand{\intervalalg}[3]{{\sc{Checkpoint}(\ensuremath{#1,#2,#3})}\xspace} 
\newcommand{\levelalg}[2]{{\sc{Level}(\ensuremath{#1,#2})}\xspace} 
\newcommand{\clsmethod}[2]{{\sc{#1}}(\ensuremath{#2})\xspace} 
\newcommand{\listmethod}[1]{{\sc{List}}\ensuremath{<}{\sc{#1}}\ensuremath{>}}

\newcommand{\akhil}[1]{
    \ifdraft{
        \textcolor{red}{Akhil: {#1}}
    }\fi
}
\newcommand{\aniket}[1]{
    \ifdraft{
        \textcolor{red}{Aniket: {#1}}
    }\fi
}
\newcommand{\mike}[1]{
    \ifdraft{
        \textcolor{green}{Mike: {#1}}
    }\fi
}
\newcommand{\saurabh}[1]{
    \ifdraft{
        \textcolor{red}{Saurabh: {#1}}
    }\fi
}
\newcommand{\adithya}[1]{
    \ifdraft{
        \textcolor{blue}{Adithya: {#1}}
    }\fi
}

\newcommand{\revisionchanges}[1]{#1}

\definecolor{yescolor}{HTML}{026378}
\def\yes{\textcolor{yescolor}{\checkmark}}
\def\no{\textcolor{red}{\ding{55}}}

\ifdraft{}
\newcommand{\changed}[1]{{\color{red} #1}}
\else
\newcommand{\changed}[1]{#1}
\fi

%

\begin{abstract}


Agreement protocols are crucial in various emerging applications, spanning from distributed (blockchains) oracles to fault-tolerant cyber-physical systems. In scenarios where sensor/oracle nodes measure a common source, maintaining output within the convex range of correct inputs, known as convex validity, is imperative. Present asynchronous convex agreement protocols employ either randomization, incurring substantial computation overhead, or approximate agreement techniques, leading to high $\mathcal{\tilde{O}}(n^3)$ communication for an $n$-node system.

This paper introduces Delphi, a deterministic protocol with $\mathcal{\tilde{O}}(n^2)$ communication and minimal computation overhead. Delphi assumes that honest inputs are bounded, except with negligible probability, and integrates agreement primitives from literature with a novel weighted averaging technique. Experimental results highlight Delphi's superior performance, showcasing a significantly lower latency compared to state-of-the-art protocols. Specifically, for an $n=160$-node system, Delphi achieves an 8x and 3x improvement in latency within CPS and AWS environments, respectively.

\end{abstract}

\section{Introduction}\label{sec:intro}

In the distributed oracle problem setting, a set of $n$ nodes each with an input coming from its sensor device, aim to closely agree on a common output while tolerating a minority of Byzantine faults. 
This problem is getting increased interest recently as it captures many scenarios including fault-tolerant distributed cyber-physical systems (CPS) ~\cite{park2017fault,zhao2020blockchain,li2019eventcps,li2022byzantine,el2021collaborative}, oracle networks ~\cite{breidenbach2021chainlink,kate2023dora}, or agreement on physical sensor inputs such as the average global temperature. 
Such systems are typically deployed in compute-starved environments supported by paltry infrastructure and asynchronous networks with unpredictable delays. 
In this context, to ensure that it accurately represents the real-world variable, the output of an ideal system is close to the range of honest inputs. 
The convex agreement problem~\cite{andrei2023convexconsensus} captures this requirement by ensuring that the output is in the convex hull (or min-max range) of honest inputs.

Current solutions to the convex agreement problem can be categorized into two different classes.
In the first, protocols that use randomization using common coins~\cite{ben1983another}. 
Generally, such coins are computationally expensive and require complex cryptographic setup assumptions. 
As an example, the most efficient and popular implementation of a common coin~\cite{boneh2003aggregate} requires \bigo{n} bilinear pairing computations per coin~\cite{cachin2000constantinople,boldyreva2003threshsigs}. 
Each pairing requires $1000$ times more time and energy than symmetric key primitives used in TLS.
Furthermore, it also requires a threshold setup, which involves running an expensive Asynchronous Distributed Key Generation (ADKG) procedure amongst nodes~\cite{kokoris2019adkg,das2022practical}. 

In the second, protocols use asynchronous approximate agreement (AAA) ~\cite{dolev1986reaching}. 
These AAA protocols are deterministic and, therefore, avoid the expensive computational cost of generating common coins. 
The state-of-the-art approximate agreement protocols~\cite{abraham2004approxagreement} incurs high communication costs of \bigo{n^3} bits per round and take \bigo{\log(\frac{\honestrange}{\epsilon})} rounds to achieve outputs that are $\epsilon$-close to each other when the range of honest inputs is $\honestrange$.

Towards addressing these inefficiencies, we make a key observation regarding input from honest parties for the oracle problems. 
In the CPS settings, the ambient noise causing the difference between inputs follows distributions like Normal and Lognormal~\cite{xiao2005scheme,xiao2007distributed}. Even prices of stocks and cryptocurrencies are often modeled using distributions like Pareto and Loggamma, where samples are mostly concentrated within a small distance~\cite{kate2023dora}.
As a result, given that honest parties often measure inputs that are close to each other,
we assume that honest inputs of size $\ell$ are sampled from a \emph{thin-tail} distribution.

Building on the above observations and the corresponding assumption, 
we present a protocol that is efficient in both computation and communication.
%
%
We introduce \name{}, a signature-free deterministic protocol with \bigo{\ell n^2\log(\statparam{}\log{n})} communication, where $\statparam{}$ is a statistical security parameter. 
\name's outputs are $\epsilon$-close and are at most within a distance $\delta = O(\epsilon)$ from the convex range of honest inputs. 
See \cref{tab:relwork_combined} for a comparison to prior works. 


\begin{table*}
	\footnotesize
	\renewcommand{\thefootnote}{\fnsymbol{footnote}}
	\def\footnoteCC{\footnotemark[2]}
	\def\footnoteValidity{\footnotemark[4]}
	\def\footnoteDelphiBID{\footnotemark[5]}
	\def\footnoteAppxAggr{\footnotemark[6]}
	\def\footnoteDora{\footnotemark[1]}
	\def\footnoteFin{\footnotemark[3]}
	\centering
	\caption[Related work comparison]{\small Comparison of relevant Asynchronous Convex BA protocols
	}%
	\label{tab:relwork_combined}
	\newcolumntype{s}{>{\centering\hsize=.25\hsize}X}
	\newcolumntype{Y}{>{\centering\arraybackslash}X}
	\newcolumntype{Z}{>{\hsize=1.2\hsize}X}
	\newcolumntype{d}{>{\centering\hsize=.5\hsize}X}
	\newcolumntype{k}{>{\centering\hsize=1.2\hsize}X}
	\newcolumntype{M}{>{\centering\arraybackslash\hsize=1.9\hsize}X}
	\centering
\begin{tabularx}{\linewidth}[!ht]{Z M k d d d k Y}
		\toprule
		\multirow{2}{*}{\textbf{Protocol}}
        &
		\textbf{Communication}
		&
		\multirow{2}{*}{\textbf{Rounds}}
		&
		\multicolumn{2}{c}{\textbf{Computation}}
		&
		\textbf{Agreement}
		&
		\multirow{2}{*}{\textbf{Validity}}
		&
		\multirow{2}{*}{\textbf{Setup}}
		\\
        \cline{4-5}
		&
		\revisionchanges{\textbf{Complexity (in bits)}}
		&
		
		&
		\multicolumn{1}{c}{\textbf{Sign}}
		&
		\multicolumn{1}{c}{\textbf{Verf}}
		&
		\textbf{Distance}
		&
		
		&
		
		\\
		\midrule
		HoneyBadgerBFT~\cite{miller2016honey}
		&
		\bigo{ln^3}
		&
		\bigo{\log(n)}
		&
		\bigo{n}
		&
		\bigo{n^2}
		&
		0
		&
		$[m,M]$
		&
		DKG
		\\
		Dumbo2~\cite{guo2020dumbo}
		&
		\bigo{ln^2+\secparam{}n^3}
		&
		\bigo{1}
		&
		\bigo{n}
		&
		\bigo{n^2}
		&
		0
		&
		$[m,M]$
		&
		HT-DKG
		\\
		FIN~\cite{duan2023sigfreeacs}
		&
		\bigo{ln^2+\secparam{} n^3}
		&
		\bigo{1}
		&
		\bigo{\log(n)}
		&
		\bigo{n\log(n)}
		&
		0
		&
		$[m,M]$
		&
		DKG
		\\
		WaterBear~\cite{zhang2022waterbear}
		&
		\bigo{ln^3 + exp(n)}
		&
		\bigo{exp(n)}
		&
		0
		&
		0
		&
		0
		&
		$[m,M]$
		&
		Auth.\ channels
		\\
		Abraham \etal~\cite{abraham2004approxagreement}\footnoteAppxAggr{}
		&
		\bigo{ln^3\log(\frac{\honestrange}{\epsilon})+n^4}
		&
		\bigo{\log(\frac{\honestrange}{\epsilon})}
		&
		0
		&
		0
		&
		$\epsilon$
		&
		$[m,M]$
		&
		Auth.\ channels
		\\
		\midrule
		\multirow{2}{*}{\textbf{\protocol}\footnoteDelphiBID\footnoteAppxAggr{}}
		&
		\bigo{ln^2\frac{\honestrange}{\epsilon}(\log(\frac{\honestrange}{\epsilon}\log{\frac{\honestrange}{\epsilon}})+\log(\statparam{}\log{n}))}
		&
		\bigo{\log(\frac{\honestrange}{\epsilon}\log{\frac{\honestrange}{\epsilon}})+\log(\statparam{}\log{n})}
		&
		\multirow{2}{*}{0}
		&
		\multirow{2}{*}{0}
		&
		\multirow{2}{*}{$\epsilon$}
		&
		\multirow{2}{*}{$[m-\honestrange,M+\honestrange]$}
		&
		\multirow{2}{*}{Auth.\ channels}
		\\
		\bottomrule
		\end{tabularx}
\begin{flushleft}
	{\small
		 \honestset{} is the set of all honest inputs, $l$ is the size of each input, $m = \min(\honestset), M=\max(\honestset)$, $\honestrange = M-m$ is the range of honest inputs,$\epsilon$ is the range of outputs, and $\secparam{}$ is the cryptographic security parameter. $l< \secparam{}$ in practice.\ \textbf{Validity}: The protocol's output \oracleoutput{} is guaranteed to be within this range.\ \footnoteAppxAggr{} \textbf{Agreement Distance} These protocols are Approximate Agreement protocols, where honest outputs have a range at most $\epsilon$, where $\vert\oracleoutput{i}-\oracleoutput{j}\vert\leq\epsilon$.\ \footnoteDelphiBID{}\textbf{Probability distribution} \protocol{} requires the input range $\honestrange\leq\maxrange$, where \maxrange is an upper bound. \protocol{} utilizes the fact that honest inputs come from a distribution, and sets $\maxrange = f(n,\statparam{})$ such that $\honestrange\leq\maxrange$ with probability negligible in \statparam{}.\ We report \protocol{}'s communication complexity in the case where honest inputs are from a Normal distribution, where $f(n,\statparam{}) = \bigo{\statparam{}\log{n}}$.\ For Lognormal and Pareto distributions, $f(n,\statparam{}) = \bigo{\statparam{}n}, \bigo{2^{\statparam{}}n}$, where \protocol{} increases as \bigo{ln^2\log{\statparam{}n}} and \bigo{l\statparam{}n^2\log{\statparam{}n}}~\cite{embrechts2013modelling}.
	}
\end{flushleft}
\end{table*}

\subsection{Solution Overview}

Our starting point is a protocol \binaryaa{} that achieves \approxconsensus{} with binary inputs, i.e. either $0$ or $1$. 
In this case, we observe that the parties output values within $\epsilon$ distance with \bigo{n^2 \log{\frac{1}{\epsilon}}} communication.

Next, we divide the real input space into intervals of size $\overshoot{}$. 
Then, for each interval, each honest node $i$ participates in \binaryaa{} with input $1$ if its input \oracleinput{i} is within \overshoot{}-distance to the interval, or $0$ otherwise. 
Note that if all honest inputs are within \overshoot{} distance, there exists an interval where all honest nodes output $1$ (since all honest nodes input $1$ to that interval). 
Moreover, all intervals that are further than distance \overshoot{} to any honest input will have output $0$. 

Next, the nodes perform a weighted average among the intervals, where each interval has a representative value called \textit{checkpoint} (e.g. the midpoint value of the interval), and the output of \binaryaa{} for each corresponding interval is the weight. 
We observe that only intervals that are \overshoot{}-close to an honest input can have a non-zero output, and the overall average remains \overshoot{}-close to the honest range. 
Moreover, the weights among honest nodes remain $\epsilon$-close, ensuring that the weighted averages are also close.

This technique fails when honest nodes' inputs are further than \overshoot{} distance, in which case all weights are $0$. 
We address this by running our technique over multiple levels, each with increasing values of \overshoot{l}. 
The starting level $l=0$ has $\overshoot{0}=\epsilon$, and higher levels $l$ have $\overshoot{l} = 2^l\overshoot{0}$. 
We then introduce a novel multi-level weighted average scheme and design levels and their corresponding \overshoot{l} values based on the input distribution so that the averages are close when the honest inputs are close.
Finally, we ensure that levels with high \overshoot{l} values do not contribute to the average 
, while ensuring that our protocol incurs $\mathcal{\tilde{O}}(n^2)$ bits of communication per round.

\paragraph{Data Analysis and Performance Evaluation}
We evaluate \protocol{} in two applications: (a) A network of oracle nodes reporting the price of Bitcoin, and (b) A distributed system of surveillance Drones detecting and locating unauthorized vehicles in an area using an object detection program. 
We analyze the nodes' inputs in each application by collecting data over a prolonged period and configure \protocol{} based on this analysis. We indeed observe that measurements can be best explained using loggamma and gamma distributions in the oracle network and object detection settings, both of which are indeed thin-tailed distributions. 
We also conduct experiments in two testbeds: (a) A geo-distributed testbed on AWS for the oracle network application and (b) A distributed CPS testbed comprised of Raspberry Pi devices for Drone-based object detection. 
For $n=160$, \protocol{} takes $\frac{1}{3}$rd and $\frac{1}{8}$th the time taken by FIN~\cite{duan2023sigfreeacs}, the State-of-the-art Asynchronous Convex BA protocol, and $\frac{1}{6}$th and $\frac{1}{8}$th the time taken by Abraham \etal~\cite{abraham2004approxagreement}, the best AAA protocol, in the AWS and CPS testbed, respectively. 
We also analyze the practical implications of the Validity relaxation of \protocol{}.

\section{Preliminaries}
\subsection{System Model}
We assume a system of $n$ nodes $\honestnodes{}:= \set{1,\ldots,n}$ connected by pairwise authenticated channels and an asynchronous network.
We consider an adaptive adversary \adv, who can corrupt $\faults<\frac{n}{3}$ nodes at any time during the execution of the protocol.
\adv{} also controls the network between honest nodes and can arbitrarily delay and reorder messages but cannot drop them.
We consider a node honest if it is never faulty. 

Nodes in the system measure a physical state variable \variable and aim to agree on a value representative of \variable. A few examples of such state variables are the price of a cryptocurrency like Bitcoin, the physical location of objects like cars, and the temperature in a given area. Each node measures \variable using an on-board data source, which provides the node with a value \oracleinput{i}, an estimate of \variable. Such data sources include the price feeds provided by currency exchanges, cameras indicating the position of an object, and sensors sensing temperature. These sources have been known to have a finite accuracy and hence measure \variable with an accuracy error, conventionally modeled using a probability distribution. In line with this error, we assume node $i$'s input \oracleinput{i} is independently sampled from a random variable $X$, with a probability distribution $p(x)$. Nodes and data sources do not know any details about $p$. We model the inputs \oracleinput{i} as floating point numbers with a finite precision. We also assume $\oracleinput{i}\in[\levelstart,\levelend]$, where \levelstart{} and \levelend{} are very small and large finite numbers, respectively, defined at a system level.

\subsection{Problem Definition}
We define the Approximate Agreement for oracles.
\begin{definition}%
  \label{def:ba_definition}
  We denote the set of honest nodes' input values by \honestset. A protocol $\Pi_{\mathsf{AA}}$ for $n$ nodes where node $i$ inputs \oracleinput{i} and outputs \oracleoutput{i} solves \approxconsensus with $\rho$-relaxed validity if it satisfies the following properties. 
  \begin{enumerate}
	\item \textbf{Termination}: Each honest node must eventually produce an output \oracleoutput{i}.
	\item \textbf{\overshoot{}-relaxed Min-Max Validity}:  Each honest node's output must be within the \overshoot{}-relaxed interval formed by honest inputs. Let the interval formed by honest inputs is $[m,M]$, where $m = \min(\honestset)$ and $M = \max(\honestset)$. For every honest node $i\in\honestnodes$:
	$ m-\overshoot{} \le \oracleoutput{i}  \le M+\overshoot{} $
	\item \textbf{$\epsilon$-agreement}: The outputs of any pair of honest nodes $i$ and $j$ are within $\epsilon$ of each other, i.e.,
	$\left\vert\aadecvalue{i} - \aadecvalue{j}\right\vert < \epsilon$. $\epsilon$ is called the agreement distance. 
  \end{enumerate}
\end{definition}

\subsection{Approximate Agreement for Binary Inputs}
Our starting building block \binaryaa{} is an Approximate Agreement protocol with convex validity (i.e., $0$-relaxed validity) for binary inputs (every honest node has input $0$ or $1$), and communicating \bigo{n^2} bits per round, similar to the notion of binary Proxcensus by Ghinea, Goyal and Liu-Zhang \cite{ghinea2022round}.

Our protocol builds upon a weaker variant of the Binary Value broadcast primitive defined in Mostefaoui, Moumen, and Raynal~\cite{mostefaoui2015signature}. This primitive can be instantiated with (a straightforward adaptation of the) Crusader Agreement protocol by Abraham, Ben-David, and Yandamuri~\cite{abraham2022bca}, with \bigo{n^2} bits per round, and terminating in three rounds.
\begin{definition}\label{def:wbv_broadcast}
	A protocol $\Pi_{\mathsf{BV}}$ for $n$ nodes where node $i$ inputs a value \oracleinput{i} and outputs a set of values \bvoutputset{i} achieves weak Binary Value broadcast if it satisfies the following properties.
	\begin{itemize}
		\item \textbf{Termination}: All honest nodes eventually terminate and output a non-empty set \bvoutputset{i}, where $\sizeof{\bvoutputset{i}}\geq 1$.
		
		\item \textbf{Justification}: If the output set \bvoutputset{i} of an honest node contains value $v$, then $v$ must have been the input of at least one honest node. 
		
		\item \textbf{Weak Uniformity}: The output sets of any pair of honest nodes $i$ and $j$ must have a non-empty intersection. 
		\[ \bvoutputset{i}\cap\bvoutputset{j}\neq\emptyset \]
	\end{itemize}
\end{definition}
\begin{algorithm}
	\small
	\caption{\revisionchanges{Binary Approximate Agreement (\binaryaa) protocol}}
	\begin{algorithmic}[1]
		\State{} \revisionchanges{\textbf{INPUT}: \nodeset, $\bitinput{i}\in\{0,1\}$, $\epsilon$}
		\HorizontalLine{}
		\Statex{}{\color{myorange}\(\triangleright\) \textrm{Parameter Setup}}
		\LineComment{\roundmax is the number of rounds to run and \ensuremath{r} is the current round}
		\State{} \revisionchanges{$\roundmax \gets \log_2(\frac{1}{\epsilon}); r = 1$}
		\LineComment{\echolist{}{} and \echotwolist{}{} stand for list of ECHO1 and ECHO2 messages received, \bvoutputset{i} is the output set for weak BV broadcast}
		\State{} \revisionchanges{$\echolist{i}{}[r]\gets\{\}; \echotwolist{i}{}[r]\gets\{\};\bvoutputset{i}[r]\gets\{\}$ for $r\in\{1,2,\hdots,\roundmax\}$}
		\HorizontalLine{}
		\Statex{}{\color{myorange}\(\triangleright\) \textrm{Begin protocol}}
		\On{\revisionchanges{receiving \msg{\text{Start,ID.i},r=1}}}\label{algline:new_round_binaa}
		\State{}\revisionchanges{\bitinput{i,1}$\gets\bitinput{1}$}
		\State{}\revisionchanges{$\echolist{i}{}[r]\gets\echolist{i}{}[r]\cup\{(\bitinput{i,1},i)\}$}
		\State{}\revisionchanges{\sendall{\msg{\text{ECHO1},\bitinput{i},r}}}\label{algline:round_init} \Comment{Send ECHO1 message}
		\EndOn{}
		\On{\revisionchanges{receiving \msg{\text{ECHO1},\bitinput{j,r},r} from node $j$}}
		\State{}\revisionchanges{$\echolist{i}{}[r]\gets\echolist{i}{}[r]\cup\{(\bitinput{j,r},j)\}$}
		\EndOn{}
		\On{\revisionchanges{receiving ECHO1s from \minthreshold nodes for a value $b$ in current round $r$ i.e. there exists $b$ such that at least \minthreshold $(b,k): k\in\nodeset$ values are in $\echolist{i}{}[r]$}}
		\State{}\revisionchanges{\sendall{\msg{\text{ECHO1},b,r}}} \Comment{Bracha amplify value $b$}\label{algline:bin_aa:amplify}
		\EndOn{}
		\On{\revisionchanges{receiving ECHO1s from \highthreshold nodes for a value $b$ and not previously sending an ECHO2 message in current round $r$}}
		\State{}\revisionchanges{$\echotwolist{i}{}[r]\gets\echotwolist{i}{}[r]\cup\{(b,i)\}$}
		\State{}\revisionchanges{\sendall{\msg{\text{ECHO2},b,r}}}
		\EndOn{}
		\State{}\revisionchanges{\textbf{wait until} one of the following conditions is true for round $r$:}
		\Indent
		\State{}\revisionchanges{(1) Receiving ECHO1s from \highthreshold nodes for two values $\bitinput{1},\bitinput{2}$ i.e. $\exists \bitinput{1},\bitinput{2}$ such that \highthreshold $(\bitinput{1},k): k\in\nodeset$ values and \highthreshold $(\bitinput{2},k): k\in\nodeset$ values are in the set $\echolist{i}{}[r]$}\label{algline:bin_aa_condition_1}
		\State{}\revisionchanges{(2) Receiving ECHO2s from \highthreshold nodes for a value $\bitinput{}$ i.e. $\exists \bitinput{}$ such that \highthreshold $(b,k): k\in\nodeset$ values are in the set $\echotwolist{i}{}[r]$}\label{algline:bin_aa_condition_2}
		\EndIndent
		\On{\revisionchanges{condition (1) being true}}
		\State{}\revisionchanges{$\bvoutputset{i}[r]\gets\{\bitinput{1},\bitinput{2}\}$}\Comment{End here for Weak BV Broadcast}
		\State{}\revisionchanges{$\bitinput{i,r+1}\gets\frac{\bitinput{1}+\bitinput{2}}{2}$}
		\State{}\revisionchanges{$r\gets r+1$ and \textbf{goto} $\cref{algline:new_round_binaa}$ to begin new round}
		\EndOn{}
		\On{\revisionchanges{condition (2) being true}}
		\State{}\revisionchanges{$\bvoutputset{i}[r]\gets\{b\}$}\Comment{End here for Weak BV Broadcast}
		\State{}\revisionchanges{$\bitinput{i,r+1}\gets b$}
		\State{}\revisionchanges{$r\gets r+1$ and \textbf{goto} $\cref{algline:new_round_binaa}$ to begin new round}
		\EndOn{}
		\State{}\revisionchanges{\textbf{output} the value \bitinput{i,\roundmax} after terminating \roundmax rounds} 
	\end{algorithmic}
	\label{alg:bin_aa}
\end{algorithm}

\revisionchanges{We describe the \binaryaa{} protocol in~\cref{alg:bin_aa} and give a brief intuition of its functioning here. \binaryaa{} proceeds in iterations, where in each iteration $r$, each node $i$ participates in an instance of a BV broadcast protocol $\Pi_{\mathsf{BV},r}$ with an input $z_i$ (in the initial iteration, $z_i =  \oracleinput{i}\in\{0,1\}$). Each iteration achieves the weak BV-broadcast primitive defined in~\cref{def:wbv_broadcast}.} 

\revisionchanges{We first show that the first iteration of the protocol satisfies all three properties specified in~\cref{def:wbv_broadcast}. First, the protocol satisfies Termination because honest inputs are binary, and at least one value $b$ would have been possessed by \minthreshold honest nodes. This ensures that $b$ will be echoed by enough honest nodes and eventually enable honest nodes to terminate. Second, the protocol satisfies Justification because faulty nodes can produce at most \faults ECHO1/ECHO2 messages on a value not possessed by honest nodes. Finally, the protocol satisfies weak uniformity because of two main reasons - (a) Each honest node sends at most one ECHO2 message, which implies at most one value $b$ can receive \highthreshold ECHO2s. This ensures that no pair of honest nodes satisfy condition 2 (\cref{algline:bin_aa_condition_2}) for two different values $\bitinput{1}\neq\bitinput{2}$. (b) A node that terminates this iteration by satisfying the first condition (\cref{algline:bin_aa_condition_1}) must have at least one value in common with every other honest node. This is because only values $0,1$ can receive \highthreshold ECHO1s.}

\revisionchanges{Second, at the end of each iteration, each node $i$ updates its value as the average of values in the output set $\bvoutputset{i,r}$. Given the set of honest inputs is $\{\oracleinput{0},\oracleinput{1}\}$, the set of honest output sets from $\Pi_{\mathsf{BV}}$ is either $\{\{\oracleinput{0}\}\},\{\{\oracleinput{0}\},\{\oracleinput{0},\oracleinput{1}\}\}$, $\{\{\oracleinput{1}\},\{\oracleinput{0},\oracleinput{1}\}\}$, or $\{\{\oracleinput{1}\}\}$. The range of honest inputs decreases by at least $\frac{1}{2}$ at each iteration. Further, the updated values in each iteration satisfy the binary input assumption. When run for $\log{\frac{1}{\epsilon}}$ iterations, this protocol achieves approximate agreement with communication complexity of \bigo{n^2\log^2{\frac{1}{\epsilon}}} bits.}

This communication can be reduced to \bigo{n^2\log(\frac{1}{\epsilon})\log\log(\frac{1}{\epsilon})} bits using a small modification. In the beginning of each round $r>1$, a node broadcasts a new type of message $\langle VAL,2L/L/C/R/2R, r\rangle$ (replacing the ECHO1 message on \cref{algline:round_init}), where $L/C/R$ signify whether the node's state value \bitinput{i,r} moved to the left by one or two spaces ($L,2L$), stayed at \bitinput{i,r-1} ($C$), or moved to the right by one or two spaces ($R,2R$). When a node $j$ receives a message $\langle VAL,L,r\rangle$ from node $i$, it waits for all $\langle VAL,.,r_i\rangle$ messages from rounds $r_i\in\{1,\hdots,r\}$. Based on this sequence of messages, $j$ deduces $i$'s state value $\bitinput{i,r}$. Then, it counts $i$'s $\langle VAL\rangle$ message as an ECHO1 message for value \bitinput{i,r}. Further, while amplifying other values using ECHO1 and ECHO2 messages, nodes use this technique to denote the value they are echoing. This technique of waiting for messages from prior rounds has been described in Abraham \etal~\cite{abraham2004approxagreement} as FIFO-broadcast. The FIFO broadcast primitive delivers messages in the order that they were broadcast by the sender. 

We give a brief rationale about why this technique works. We know that in each round, the range of honest state values either reduces by a $\frac{1}{2}$ fraction or collapses to $0$. Further, honest state values after each round are binary. Let $S_r = \{\bitinput{r,0},\bitinput{r,1}\}$ be the set of honest state values in round $r$. Then, the set of values in round $r+1$ becomes one of the following four sets: $S_{r+1} = \{\bitinput{r,0}\},\{\bitinput{r,0},\frac{\bitinput{r,0}+\bitinput{r,1}}{2}\},\{\frac{\bitinput{r,0}+\bitinput{r,1}}{2},\bitinput{r,1}\},\{\bitinput{r,1}\}$. We denote an honest node with value $\bitinput{r,1}$ in round $r$ updating its state to $\frac{\bitinput{r,0}+\bitinput{r,1}}{2}$ (or shifting towards the left by $\frac{1}{2^r}$) in round $r+1$ using the term $L$. Similarly, we denote an honest node updating its state from $\bitinput{r,1}$ to $\bitinput{r,0}$ (or shifting towards the left by $\frac{1}{2^{r-1}}$) using the term $2L$. We use similar notation for a node shifting its state value towards the right. Therefore, using a sequence of such state shifts, a node $j$ can correctly calculate node $i$'s state value in round $r$. 

The communication complexity of this updated protocol is \bigo{n^2\log(\frac{1}{\epsilon})\log\log(\frac{1}{\epsilon})}. The $\log\log(\frac{1}{\epsilon})$ factor is due to the inclusion of the round number in each message. 

\section{Design}

\rebox{}
\begin{figure}
\small
	\begin{flushleft}
		\begin{itemize}[wide,labelindent=0pt]
			\item \textbf{Real range \honestrange} is the real difference between maximum and minimum honest inputs $\honestrange = \max(\honestset)-\min(\honestset)$.
			\item \textbf{Max range \maxrange} is the maximum possible difference between maximum and minimum honest inputs $\honestrange\leq\maxrange$.
			\item \textbf{Separator \overshoot{l}} is the length of each interval or difference between adjacent checkpoints in a level $l$.
			\item \textbf{Checkpoint \intervalmidpoint{k}{y}{l}} or \intervalalg{\checkpointindex}{y}{\level{}} is the $k$th multiple of the separator \overshoot{l}: $\intervalmidpoint{k}{y}{l} = k\overshoot{l}$, where $k$ is an integer between $k\in[\frac{\levelstart}{\overshoot{l}},\frac{\levelend}{\overshoot{l}}]$. Level $l$ has checkpoints throughout the space of honest inputs $[\levelstart,\levelend]$, where adjacent checkpoints are separated by $\overshoot{l} = \exponent^l\overshoot{0}$.
			\item \textbf{Weight $\binaastateval{\checkpointindex}{y}{l}$} is the weight of checkpoint \intervalmidpoint{\checkpointindex}{}{l}. Nodes achieve \approxconsensus{} on these values using \binaryaa{} protocol.
			\item \textbf{Level \levelalg{\level}{\overshoot{}}} or \levelstate{l} is the level object. It contains a list of intervals \interval{x}{y}{l} at level $l$ between endpoints \levelstart{} and \levelend. $\overshoot{l} = \exponent^l\overshoot{0}$ is the distance between adjacent checkpoints.
		\end{itemize}
	\end{flushleft}
	\caption{Description of symbols used in \protocol}%
	\label{tab:symbol_desc}
\end{figure}
\unbox{}


In this section, we present the design of our protocol \protocol. We first provide context, explain the limitations of prior solutions, and then discuss our approach.

\subsection{Limitations of prior protocols}
Prior \approxconsensus protocols like Dolev \etal~\cite{dolev1986reaching} and Abraham \etal~\cite{abraham2004approxagreement} proceed in a round-based manner, where in each round, every node collects a set of values from other nodes, and updates its state using a Trimmed mean of this set. In these protocols, reaching \approxconsensus requires each node to collect $2\faults+1$ values in common with other honest nodes. This is for two major reasons. (a) Even with \faults faulty values, each honest node must ensure that its updated state is within the range of honest inputs, and (b) Even after trimming \faults greatest and smallest values in their sets, honest nodes must have enough values in common to reduce their range. 
Dolev \etal~\cite{dolev1986reaching} achieve this condition using multicasts with $n=5\faults+1$, where a faulty node can make different honest nodes accept different values. 
Achieving this at $n=3\faults+1$ resilience requires restricting equivocation by faulty nodes, which requires using the Reliable Broadcast (RBC) primitive. RBC has a lower bound of \bigo{n^2} bits, which results in an overall \bigo{n^3} communication complexity. Therefore, we observe that this requirement of collecting $2\faults+1$ values in common with other nodes, which is necessary for achieving strict convex-hull Validity, is the main reason for the \bigo{n^3} communication complexity of all prior \approxconsensus protocols. 

\subsection{Our Approach}
We overcome this \bigo{n^3} communication bottleneck by proposing \protocol{}, an \approxconsensus protocol that does not require RBC. \protocol{} uses a novel approach based on checkpoints and Binary \approxconsensus, which achieves communication efficiency by trading off a relaxation in the Validity condition, specified in~\cref{def:ba_definition}. 

\paragraph{Notation} We divide the space of possible honest inputs between $[\levelstart,\levelend]$ among checkpoints separated by a system parameter \overshoot{}. A checkpoint $\intervalmidpoint{\checkpointindex}{}{} = \checkpointindex\overshoot{}$ represents the space of values between $[\checkpointindex\overshoot{}-\frac{\overshoot{}}{2},k\overshoot{}+\frac{\overshoot{}}{2}]$, where $k$ is the set of all integers in $(\dfrac{\levelstart}{\overshoot{}},\dfrac{\levelend}{\overshoot{}})$.
Nodes run a \binaryaa{} protocol instance for all checkpoints in $[\levelstart,\levelend]$ to approximately agree on a \textit{representative weight} \binaastateweight{\checkpointindex}{y}{l} for each checkpoint.
We denote the corresponding pseudocode in~\cref{alg:n_2_BA}.
Each level, indexed by the number $l$, is characterized by the difference between two consecutive checkpoints \overshoot{l}. Further, each level comprises checkpoints in $[\levelstart,\levelend]$, where each checkpoint is a multiple of \overshoot{l}. 

For ease of understanding, we first describe a basic, single-level version of \protocol{}, explain the challenges, and then explain the multi-level version that addresses these challenges. 

\subsubsection{Single level \approxconsensus} We first describe a basic version of \protocol{} at a single level with distance between checkpoints \overshoot{}. Nodes run \binaryaa{} for all checkpoints in a level. 
An honest node inputs $v=1$ to checkpoint $\checkpointindex\overshoot{}$ only if $|\oracleinput{i}-\checkpointindex\overshoot{}|\leq\overshoot{}$, and inputs 0 otherwise (\cref{algline:n_2_BA:initiate_binaa} in~\cref{alg:n_2_BA}). Nodes then send out messages for checkpoints and handle incoming messages. Each node then checks whether a round $r$ has terminated for a checkpoint. Upon terminating round $r$ of \binaryaa{} of all checkpoints in the level, nodes initiate round $r+1$. The nodes terminate a \binaryaa{} instance after running for \roundmax rounds (specified later). Upon terminating all \binaryaa{} instances, nodes calculate a weighted average of checkpoints $\oracleoutput{i} = \dfrac{\sum_{k}(\checkpointindex\overshoot{})\times\binaastateweight{k}{}{}}{\sum_{k}\binaastateweight{k}{}{}}$ (\cref{algline:lev_wavg} in~\cref{alg:n_2_BA})\@.

\paragraph{Intuition} We observe that when honest inputs are clustered within a distance $\honestrange{}\leq\overshoot{}$, at least one checkpoint will have weight $\binaastateweight{\checkpointindex}{}{}=1$. This is because the Validity property of \binaryaa ensures that if all honest nodes input the same value to any \binaryaa{} instance \binaastate{\checkpointindex}{j+1}{}, then all nodes must output that value. This property gives way to two key observations: (a) If all honest inputs \oracleinput{i} are within \overshoot{} distance, the weight $\binaastateweight{\checkpointindex}{j+1}{}=1$ for at least one checkpoint \intervalmidpoint{\checkpointindex}{j+1}{}, and (b) No checkpoint $\intervalmidpoint{\checkpointindex}{k+1}{}:|\intervalmidpoint{\checkpointindex}{k+1}{}-\oracleinput{i}|>\overshoot{} \forall i\in\honestset$ will have a non-zero weight. The first observation guarantees approximate agreement of the weighted average, whereas the second observation enables us to show that the final output will strictly lie between the \overshoot{}-relaxed interval of honest inputs $[\min(\honestset)-\overshoot{},\max(\honestset)+\overshoot{}]$.

\paragraph{Challenges} The described approach produces a valid result only when honest nodes' inputs are within $\honestrange\leq\overshoot{}$ distance. When $\honestrange>\overshoot{}$, the weights of all checkpoints can be zero, resulting in a division by zero. Moreover, estimating the range \honestrange{} in a distributed manner also costs a whopping \bigo{n^4} complexity in prior protocols~\cite{abraham2004approxagreement}. We explain this challenge pictorially in \cref{fig:singlelevel}. 

We assume an upper bound \maxrange{} on the range \honestrange, and utilize it to ensure that nodes always output a well-defined weighted average. Prior works like Chakka \etal~\cite{kate2023dora} also utilized this upperbound assumption to achieve BA efficiently, and practically justified this assumption~\cite{kate2023dora}. Moreover, as many applications are modeled using thin-tailed probability distributions, this assumption can be replaced with a statistical security parameter \statparam{}. In~\cref{sec:analysis}, we show that $\maxrange=\bigo{\statparam{}\honestrange_\mathsf{mean}}$ is sufficient to ensure that $\honestrange\leq\maxrange$ with probability $1-\negl(\statparam{})$. We can statically set $\overshoot{} = \maxrange$ to ensure the honest inputs are clustered within the \overshoot{} range, where $\honestrange\leq\overshoot{}$ even in the worst case.

However, this \overshoot{} causes an abnormally high Validity relaxation of \maxrange, even in the average case with much lesser \honestrange. For example, as we study in~\cref{sec:eval}, in the case of nodes agreeing on the price of a cryptocurrency, honest nodes' inputs can differ by hundreds of dollars in situations of drastic volatility, which prompts us to assume $\maxrange\sim 100\$ $, and set $\overshoot{}\sim 100\$$ to guarantee termination. This $\overshoot{}$ value causes a high Validity relaxation of $100\$ $, even in the average case when honest nodes input values close to the ground truth and close to each other with $\honestrange\sim 10\$ $. This high Validity relaxation generates inaccurate results, which misrepresent the ground truth.
\begin{figure}
	\centering
	\includegraphics[width=0.75\linewidth]{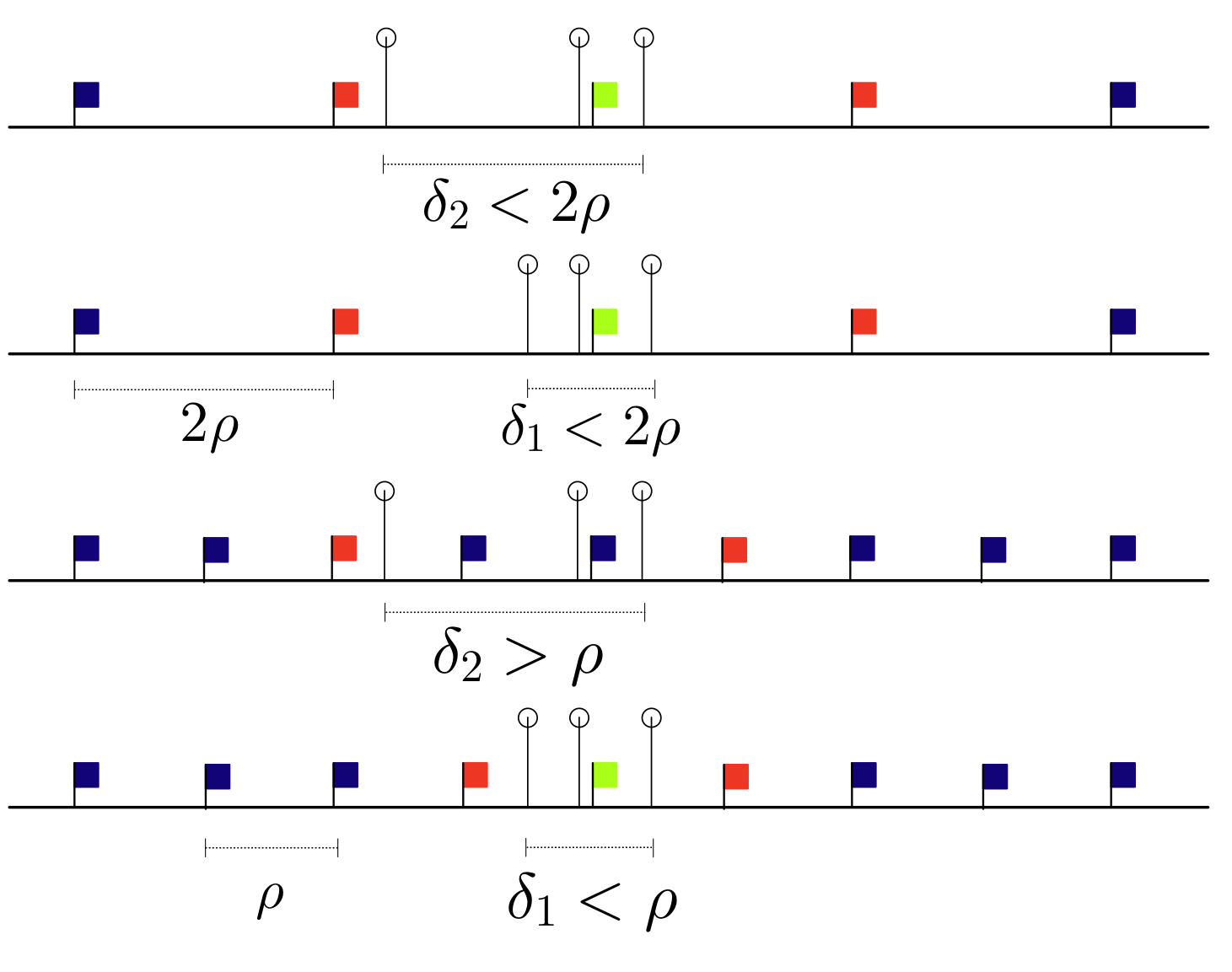}
	\caption{\small\textbf{\protocol{} with one level}: The flags denote the checkpoints, and antennae denote honest inputs. Green checkpoints have weight $\binaastateweight{\checkpointindex}{}{}=1$ and drive agreement among nodes, and red checkpoints are outside the range of honest inputs that can have a non-zero weight and contribute to Validity relaxation. When $\honestrange_2>\overshoot{}$, no green checkpoint exists, which results in agreement failure. A higher $2\overshoot{}$ enables nodes to reach agreement with range $\honestrange_2$, but also adds the weight of farther red checkpoints when the range is $\honestrange_1$, which affects Validity negatively.}
	\label{fig:singlelevel}
\end{figure}

\begin{algorithm}
	\small
	\caption{\protocol{} protocol}
		\begin{algorithmic}[1]
			\State{}\textbf{INPUT}: $\oracleinput{i}\in[\levelstart,\levelend]$,\overshoot{0},$\maxrange$,$\epsilon$
			\HorizontalLine{}
			\Statex{}{\color{myorange}\(\triangleright\) \textrm{Setup}}
			\LineComment{Level and round parameters}
			\State{}$\levelmax \gets \log_\exponent(\dfrac{\maxrange}{\overshoot{0}}); \aaprecision \gets \dfrac{\epsilon}{4\maxrange\levelmax n};\roundmax \gets \log_2(\dfrac{1}{\aaprecision})$
			\For{\revisionchanges{$l$ in $\{0,1,\hdots,\levelmax\}$}}\label{algline:n_2_BA:init_levels}
				\State{}\revisionchanges{\textbf{Define }$\overshoot{l} = 2^l\overshoot{0}$}
				\For{\revisionchanges{$k$ in $\{i\in \mathbb{Z}: \lceil\dfrac{\levelstart}{\overshoot{l}}\rceil\leq i\leq \lfloor\dfrac{\levelend}{\overshoot{l}}\rfloor\}$}}
				\LineComment{Initialize Binary AA instances for all checkpoints}
				\State{} \revisionchanges{$\binaastate{k}{}{l}\gets$\binaryaa()}
				\EndFor{}
			\EndFor{}
			\Statex{}{\color{myorange}\(\triangleright\) \textrm{\binaryaa phase}}
			\On{receiving \msg{Start,ID.i}}
			\State{}$\levelencap\gets\{\}$
			\For{$l$ in $\{0,1,\hdots,\levelmax\}$}
			\LineComment{Checkpoints separated by \overshoot{l}}
			\State{}$\intervalmidpoint{\closestcc-1}{}{l},\intervalmidpoint{\closestcc}{c+1}{l}\gets$Two closest checkpoints to input \oracleinput{i} at lev $l$
			\State{}\revisionchanges{\textbf{Start} \binaryaa instances \binaastate{\closestcc-1}{}{l} and \binaastate{\closestcc}{}{l} with input $\bitinput{i,\closestcc-1}^l=1, \bitinput{i,\closestcc}^l=1$ and all other instances with input $\bitinput{i}=0$.} \label{algline:n_2_BA:initiate_binaa}
			\EndFor{}
			\State{}\revisionchanges{\textbf{Bundle} messages of all \binaryaa instances together and invoke \textbf{SendAll} on the bundled message}
			\EndOn{}
			\Statex{}{\color{myorange}\(\triangleright\) \textrm{Aggregation phase}}
			\On{\revisionchanges{terminating \roundmax rounds for all \binaryaa instances at all levels:}}
			\LineComment{1. Aggregate weights of checkpoints at each level. Each level has a representative value \levelval{l} and a weight \levelweight{l}.}
			\For{$l$ in $\{0,1,\hdots,\levelmax\}$}
			\If{\revisionchanges{$\exists$ a \binaryaa instance \binaastate{\checkpointindex}{}{l} with output $>0$}}
			\State{}\revisionchanges{\binaastateweight{\checkpointindex}{y}{l}$\gets \binaastate{k}{}{l}.$\clsmethod{Output}{}}
			\State{}\revisionchanges{\intervalmidpoint{\checkpointindex}{}{l}$\gets \checkpointindex\overshoot{l}$}
			\State{}\begin{multline*}
				\revisionchanges{(\levelval{l},\levelweight{l})\defeq \left(\begin{aligned}
					\dfrac{\sum_{\lceil\frac{\levelstart}{\overshoot{l}}\rceil\leq k\leq\lfloor\frac{\levelend}{\overshoot{l}}\rfloor}\binaastateweight{\checkpointindex}{y}{l}\intervalmidpoint{\checkpointindex}{y}{l}}{\sum_{\lceil\frac{\levelstart}{\overshoot{l}}\rceil\leq k\leq\lfloor\frac{\levelend}{\overshoot{l}}\rfloor}\binaastateweight{\checkpointindex}{y}{l}},
					\max_{\lceil\frac{\levelstart}{\overshoot{l}}\rceil\leq k\leq\lfloor\frac{\levelend}{\overshoot{l}}\rfloor}{\binaastateweight{\checkpointindex}{y}{l}}
				\end{aligned}\right)}
			\end{multline*}\label{algline:lev_wavg}
			\Else{}
			\LineComment{As weighted average is undefined when all weights are 0, we handle this case by assigning a custom weight for level $l$}
			\State{}\revisionchanges{$(\levelval{l},\levelweight{l}):=(\oracleinput{i},\aaprecision)$}
			\EndIf{}
			\EndFor{}
			
			\LineComment{2. Cross-level Aggregation: Aggregate weights across levels}
			\State{}\revisionchanges{$\lpweight{0}\gets\levelweight{0}^2$\label{algline:lp_weight}}
			\For{\revisionchanges{$l$ in $\{1,2,\hdots,\levelmax\}$}}
			\State{}\revisionchanges{$\lpweight{l}\gets\levelweight{l}\times\sizeof{\levelweight{l}-\levelweight{l-1}}$}
			\EndFor{}
			\State{}\revisionchanges{\oracleoutput{i} \defeq{} $\dfrac{\sum_{l=0}^{\levelmax}(\lpweight{l}\times\levelval{l})}{\sum_{l=0}^{\levelmax}\lpweight{l}}$}\label{algline:level_wavg}
			\State{}\revisionchanges{\textbf{output } \oracleoutput{i}}
			\EndOn{}
		\end{algorithmic}%
	\label{alg:n_2_BA}
\end{algorithm}


\begin{figure}[!th]
	\centering
	\includegraphics[width=0.75\linewidth]{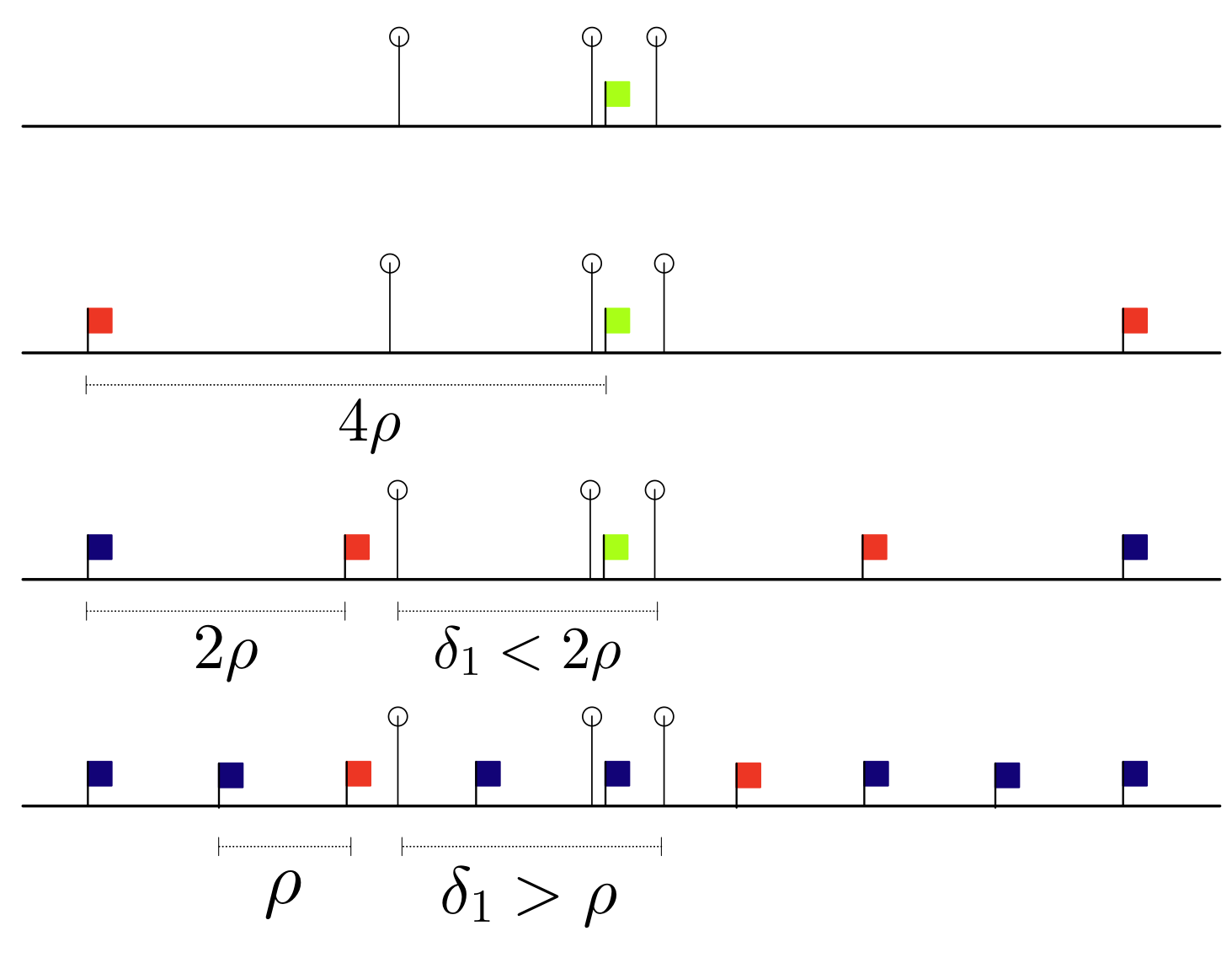}
	\caption{\small\textbf{\protocol{} with multiple levels}: As $\honestrange{}<2\overshoot{}$, all levels $l\geq2$ have at least one green flag, which implies $\levelweight{l}=\max_{k}\binaastateweight{\checkpointindex}{l}{}$ is 1 for levels 2,3,4. We eliminate contributions of levels 3 and 4 by using weight $\lpweight{l}=\levelweight{l}|\levelweight{l}-\levelweight{l-1}|$ in the weighted average.}
	\label{fig:multlevel}
\end{figure}

\subsubsection{Multi-level Approximate Agreement} 
We address this dependence of Validity relaxation on the maximum range \maxrange{} by running \binaryaa{} across checkpoints over multiple levels. We describe this protocol, \protocol{}, in~\cref{alg:n_2_BA}. We define a level $\level{}$ where checkpoints are separated by the level-specific separator distance $\overshoot{l} = \exponent^{l}\overshoot{0}$ (\cref{algline:n_2_BA:init_levels} in~\cref{alg:n_2_BA}).\
\overshoot{0} is the separator at the starting level $0$. Each node inputs 1 to instance \binaastate{\checkpointindex}{j+1}{l} iff its input \oracleinput{i} is within distance $\exponent^l\overshoot{0}$ from the checkpoint $\intervalmidpoint{\checkpointindex}{j+1}{l}$ (\cref{algline:n_2_BA:init_levels} in~\cref{alg:n_2_BA}). Upon terminating \roundmax{} rounds of \binaryaa{} instances, we calculate the weighted average of checkpoints using the agreed-upon weights (\cref{algline:lev_wavg} in~\cref{alg:n_2_BA}).

\paragraph{Intuition} We give a pictorial description of this approach in \cref{fig:multlevel}. Since we know that $\honestrange\leq\maxrange$, there must exist a level $\critlevel\leq\levelmax$ such that every level $\geq \critlevel$ must have at least one value $\intervalmidpoint{\checkpointindex}{j+1}{l}$ with weight $\binaastateval{\checkpointindex}{j+1}{l}=1$. In cases with relatively small \honestrange, the level \critlevel{} is small. As the checkpoints at higher levels contribute a larger Validity relaxation to the output, we aim to minimize their contribution to the weighted average. We utilize the fact that the maximum weight of a checkpoint in a level $\max{\binaastateweight{\checkpointindex}{}{l}} = 1$ for any level $l\geq\critlevel$. We set the weight of a level $\lpweight{l} = \levelweight{l}|\levelweight{l}-\levelweight{l-1}|$, where \levelweight{l} is the maximum weight of a checkpoint in level $l$ (\cref{algline:lp_weight} in~\cref{alg:n_2_BA}). This operation is analogous to a differentiation. Any checkpoint at level $l>\critlevel$ will not contribute to the output.

\subsection{Optimizing Communication}\label{subsec:optimize_comm}
\protocol{} requires running the Binary Approximate Agreement protocol for all intervals between \levelstart{} and \levelend{}. However, separately sending messages for each \binaryaa{} instance would require $(\dfrac{\levelend-\levelstart}{\overshoot{0}})n^2$ messages per round, which is infeasible in practice. We bundle messages from multiple \binaryaa{} instances and treat them as a single \binaryaa{} instance. For example, we consider a level $l$ with honest inputs between checkpoints $\checkpointindex_1,\checkpointindex_2$. All honest nodes will input 0 to checkpoints from $[\levelstart,\checkpointindex_1]$ and $[\checkpointindex_2,\levelend]$. This implies all these checkpoints require only \bigo{n^2} communication per round. Further, for cases when $\checkpointindex_2-\checkpointindex_1 = \bigo{c}$, the total communication per level is $\bigo{n^2}$ bits per round. We also perform cross-level optimization. If all honest inputs are within two consecutive checkpoints \intervalmidpoint{\checkpointindex}{}{\critlevel} and \intervalmidpoint{\checkpointindex+1}{}{\critlevel} at level \critlevel, the \binaryaa{} data for levels \critlevel{} to \levelmax{} can be represented with constant bits. 
\section{Analysis}\label{sec:analysis}
\baselineskip=12.1pt
\revisionchanges{\paragraph{Overview} We know that nodes approximately agree on the weight of each checkpoint, which implies each node’s weights are $\aaprecision$-close for all checkpoints at all levels. The weighted average function enables nodes to combine the weights and checkpoints by offering two key properties: (a) It forces the final output to be within checkpoints with non-zero weight outputs, and (b) As the weights are $\aaprecision$-close, the outputs from the function are $\epsilon$-close, which guarantees approximate agreement with Validity. However, this approach also has a caveat: it becomes undefined when the sum of weights is zero. Hence, we first ensure that the denominator in our weighted average is always greater than $\frac{1}{2}$. Next, we utilize the properties offered by the weighted average scheme to prove approximate agreement and Validity of \protocol.}

\subsection{Termination}
We discuss \protocol's termination. We show that the weighted average in~\cref{algline:level_wavg} in~\cref{alg:n_2_BA} always has a non-zero sum of weights.

\begin{theorem}\label{thm:liveness}
	\textbf{Termination}: Given that the \binaryaa{} protocol provides Safety, Termination, and Validity, every honest node terminates \protocol{} and outputs a finite, defined value.
\end{theorem}
\begin{proof}
	We know that every honest node terminates all \binaryaa{} instances corresponding to all checkpoints in levels $l = \{0,\hdots,\levelmax\}$. Consider the denominator in~\cref{algline:level_wavg} in~\cref{alg:n_2_BA}. Every honest node starts \binaryaa{} for all intervals in every level until \levelmax. From the Termination property of \binaryaa{}, every honest node terminates \roundmax{} rounds for all checkpoints until level \levelmax{}. 
	\begin{gather*}
		S = \levelweight{0}^2 + \sum_{l=1}^{\levelmax} \levelweight{l}\times|\levelweight{l}-\levelweight{l-1}|\\
		\implies S \geq \frac{1}{2}(2\levelweight{0}^2+\sum_{l=1}^{\levelmax}2\levelweight{l}^2 -2\levelweight{l}\levelweight{l-1})\\
		\implies S \geq \frac{1}{2}(\levelweight{0}^2+\levelweight{\levelmax}^2+\sum_{l=1}^{\levelmax}{(\levelweight{l}-\levelweight{l-1})}^2)\geq \frac{\levelweight{\levelmax}^2}{2}
	\end{gather*}
    As honest inputs have a maximum difference $\honestrange\leq\maxrange$, there must exist level $l\leq\levelmax$ such that 
     $\levelweight{l}=1$, which ensures that the sum of weights is $\geq\frac{1}{2}$. 
\end{proof}
\subsection{Validity}
The Validity offered by \protocol{} depends on two factors: (i) The separator at level $0$ \overshoot{0}, and (ii) The rate of increase of interval length across levels. These factors decide the number of levels in \protocol{} and therefore influence the Validity. 

\begin{lemma}\label{lem:level_weight_bound}
	Given that the \binaryaa{} protocol satisfies Agreement, Termination, and Validity, an honest node $i$'s representative weight for level $l$ \lpweight{l,i} in~\cref{algline:lp_weight} in~\cref{alg:n_2_BA} is at most $5\aaprecision$ away from $\lpweight{l,j}\forall l\in\{0,\hdots,\levelmax\}\forall j \in \honestnodes$. Further, for all levels $l>\lceil\log(\frac{\honestrange}{\overshoot{0}})\rceil$, the weight $\lpweight{l,i}=0 \forall i\in\honestnodes$. 
	\begin{gather*}
		|\lpweight{l,i}-\lpweight{l,j}| \leq 5\aaprecision, \forall l\in\{0,\dots,\lceil\log(\frac{\honestrange}{\overshoot{0}})\rceil\}\forall \{i,j\}\in\honestnodes\\
		\lpweight{l,i} = 0 \forall l>\lceil\log(\frac{\honestrange}{\overshoot{0}})\rceil, \forall i \in \honestnodes
	\end{gather*}
\end{lemma}
\begin{proof}
	The weight \levelweight{l,i} in \lpweight{l,i} is calculated as \levelweight{l,i} = $\max\{\binaastateval{\checkpointindex}{y}{l}: \binaastate{x}{y}{l}$. 
    From agreement of \binaryaa{}, we know $|\levelweight{l,i}-\levelweight{l,j}|\leq\aaprecision \forall \{i,j\}\in\honestset$, and $|\levelweight{l,i}-\levelweight{l-1,i}|-|\levelweight{l,j}-\levelweight{l-1,j}| \leq 2\aaprecision$. . 
	
     
	\begin{multline*}
		\lpweight{l,i}-\lpweight{l,j} \leq (\levelweight{l,j}+\aaprecision)\times(|\levelweight{l,j}\\-\levelweight{l-1,j}|+2\aaprecision)-\levelweight{l,j}|\levelweight{l,j}-\levelweight{l-1,j}|
	\end{multline*}
	\begin{gather*}
		\lpweight{l,i}-\lpweight{l,j} \leq \aaprecision(2\levelweight{l,j}+|\levelweight{l,j}-\levelweight{l-1,j}|)+2\epsilon^{'2}\\
		\implies \lpweight{l,i}-\lpweight{l,j} \leq 5\aaprecision
	\end{gather*}
	
	For level $\critlevel = \lceil\log(\frac{\honestrange}{\overshoot{0}})\rceil$, the separation between checkpoints $\overshoot{l}\geq\honestrange$, implying $\levelweight{l}=1 \forall l\geq\critlevel, \lpweight{l} = 0 \forall l\geq\critlevel$. This implies there exists at least one checkpoint \intervalmidpoint{\checkpointindex}{k+1}{\critlevel} that is within a \overshoot{l} distance from all honest inputs \oracleinput{i}. From the Validity of \binaryaa, $\binaastateval{\checkpointindex}{k+1}{\critlevel}=1 \forall i\in\honestnodes, \implies \levelweight{\critlevel,i} = 1$. This statement is also true for all levels $>\critlevel$. Therefore, the weight $\lpweight{l,i} =0 \forall l>\critlevel$.
\end{proof}
\begin{theorem}\label{thm:int_validity}
	Given that the \binaryaa{} protocol satisfies Termination, and Validity, honest nodes' output values \oracleoutput{i} from \protocol{} are always within the following interval; i.e.,
		$\oracleoutput{i} \in [\min(\honestset)-\max(\overshoot{0},\honestrange),
		\max(\honestset)+\max(\overshoot{0},\honestrange)]$
\end{theorem}
\begin{proof}
	From~\cref{lem:level_weight_bound}, every level $l >\critlevel$, where $\critlevel= \lceil\log(\frac{\honestrange}{\overshoot{0}})\rceil$, has weight $\lpweight{l,i} = 0$. 
 
	First, we examine all checkpoints \intervalmidpoint{\checkpointindex}{k+1}{l} that can have a positive weight in~\cref{algline:lev_wavg} in~\cref{alg:n_2_BA} for levels $l<\critlevel$. We know that only checkpoints within a \overshoot{l} distance from an honest input \oracleinput{i} can have a positive weight. Therefore, no $\intervalmidpoint{\checkpointindex}{}{l}<\min(\honestset)-\overshoot{l}$ and $\intervalmidpoint{\checkpointindex}{}{l}>\max(\honestset)+\overshoot{l}$ can have a weight $\binaastateval{\checkpointindex}{}{l}>0$. This implies that the level's weight $\levelweight{l,i}\in(\min(\honestset)-\overshoot{l},\max(\honestset)+\overshoot{l}) \forall l<\critlevel$. The farthest point which can have a positive weight in \oracleoutput{i} until level $l<\critlevel$ is $\min(\honestset)-\honestrange$ and $\max(\honestset)+\honestrange$, because $\overshoot{l}<\honestrange, \forall l<\critlevel$. 
	
	Second, we examine checkpoints \intervalmidpoint{k}{}{l} at $l=\critlevel$, where $\overshoot{\critlevel}\in[\honestrange,2\honestrange)$. 
	There exists at least one  \intervalmidpoint{\checkpointindex}{}{l} such that $|\intervalmidpoint{\checkpointindex}{}{l}-\oracleinput{i}|\leq\overshoot{\critlevel}, \forall i\in\honestnodes$, which implies $\binaastateweight{\checkpointindex}{}{\critlevel} = 1$. The only other checkpoints that can have a non-zero weight are \intervalmidpoint{\checkpointindex-1}{}{\critlevel} and \intervalmidpoint{\checkpointindex+1}{}{\critlevel}, which ensures that level \critlevel's representative value \levelweight{\critlevel} must be within $[\frac{\intervalmidpoint{\checkpointindex-1}{}{\critlevel}+\intervalmidpoint{\checkpointindex}{}{\critlevel}}{2},\frac{\intervalmidpoint{\checkpointindex}{}{\critlevel}+\intervalmidpoint{\checkpointindex+1}{}{\critlevel}}{2}]$. As there must exist honest nodes $i,j$ such that $|\intervalmidpoint{\checkpointindex-1}{}{\critlevel}-\oracleinput{i}|<2\honestrange$ and $|\intervalmidpoint{\checkpointindex+1}{}{\critlevel}-\oracleinput{j}|<2\honestrange$, the relaxation can be at most \honestrange, which proves the theorem.
\end{proof}
\subsection{Agreement}
We prove that \protocol{} achieves $\epsilon$-agreement. 
\begin{theorem}
	Given that the \binaryaa{} protocol satisfies Agreement, Termination, and Validity, outputs of honest nodes \oracleoutput{i} after terminating \protocol{} must be within $\epsilon$ distance of each other; i.e.,
$		|\oracleoutput{i}-\oracleoutput{j}| \leq \epsilon \forall \{i,j\}\in\honestnodes$
\end{theorem}
\begin{proof}
	We consider the difference between the weighted averages of two honest nodes $O = \oracleoutput{i}-\oracleoutput{j}$. 
	\chop{
	\begin{multline*}
		O = \oracleoutput{i}-\oracleoutput{j} \\= \frac{\levelweight{0,i}^2\times\levelval{0,i}+\sum_{l=1}^{\levelmax}(\levelweight{l,i}\times|\levelweight{l,i}-\levelweight{l-1,i}|\times\levelval{l,i})}{\levelweight{0,i}^2+\sum_{l=1}^{\levelmax}\levelweight{l,i}\times|\levelweight{l,i}-\levelweight{l-1,i}|}\\-\frac{\levelweight{0,j}^2\times\levelval{0,j}+\sum_{l=1}^{\levelmax}(\levelweight{l,j}\times|\levelweight{l,j}-\levelweight{l-1,j}|\times\levelval{l,j})}{\levelweight{0,j}^2+\sum_{l=1}^{\levelmax}\levelweight{l,j}\times|\levelweight{l,j}-\levelweight{l-1,j}|}
	\end{multline*}}
	\begin{gather*}
		O = \left(\frac{\sum_{l=0}^{\levelmax}\lpweight{l,i}\levelval{l,i}}{\sum_{l=0}^{\levelmax}\lpweight{l,i}}\right)-\left(\frac{\sum_{l=0}^{\levelmax}\lpweight{l,j}\levelval{l,j}}{\sum_{l=0}^{\levelmax}\lpweight{l,j}}\right)
	\end{gather*}
	We assume $\lpweight{l,j}=\lpweight{l,i}+\adddiff{l}$. The equation then transforms to the following.
	\chop{
	\begin{gather*}
		O = \left(\dfrac{\splitdfrac{(\sum_{l=0}^{\levelmax}\lpweight{l,i}\levelval{l,i})(\sum_{l=0}^{\levelmax}\lpweight{l,i}+\adddiff{l})}{-(\sum_{l=0}^{\levelmax}(\lpweight{l,i}+\adddiff{l})\levelval{l,j})(\sum_{l=0}^{\levelmax}\lpweight{l,i})}}{(\sum_{l=0}^{\levelmax}\lpweight{l,i})(\sum_{l=0}^{\levelmax}\lpweight{l,j})}\right)\\
		O = \left(\dfrac{\splitdfrac{(\sum_{l=0}^{\levelmax}\lpweight{l,i})(\sum_{l=0}^{\levelmax}\lpweight{l,i}(\levelval{l,i}-\levelval{l,j}))}{+(\sum_{l=0}^{\levelmax}\adddiff{l})(\sum_{l=0}^{\levelmax}\lpweight{l,i}\levelval{l,i})-(\sum_{l=0}^{\levelmax}(\adddiff{l}\levelval{l,j}))}}{(\sum_{l=0}^{\levelmax}\lpweight{l,i})(\sum_{l=0}^{\levelmax}\lpweight{l,j})}\right)
	\end{gather*}}
	\begin{multline*}\label{eqn:splitprove}
		O = \left(\frac{\sum_{l=0}^{\levelmax}\lpweight{l,i}(\levelval{l,i}-\levelval{l,j})+\adddiff{l}(\levelval{avg,i}-\levelval{l,j})}{\sum_{l=0}^{\levelmax}\lpweight{l,j}}\right)
	\end{multline*}
	The term $\levelval{avg,i} = \frac{\sum_{l=0}^{\levelmax}\lpweight{l,i}\levelval{l,i}}{\sum_{l=0}^{\levelmax}\lpweight{l,i}}$. From~\cref{thm:liveness}, we know that $\sum_{l=0}^{\levelmax}\lpweight{l,j} \geq \frac{1}{2}$. Therefore, we rewrite the equation as the sum of the following two terms and then examine each term individually.
	\begin{multline*}
		O \leq \sum_{l=0}^{\levelmax}\left(\lpweight{l,i}|\levelval{l,i}-\levelval{l,j}|+\adddiff{l}|\levelval{avg,i}-\levelval{l,j}|\right)
	\end{multline*}
	We split the above sum into two different terms: $$O_1 = \sum_{l=0}^{\levelmax}\left(\lpweight{l,i}|\levelval{l,i}-\levelval{l,j}|\right)$$
	
	$$O_2=\sum_{l=0}^{\levelmax}\left(\adddiff{l}|\levelval{avg,i}-\levelval{l,j}|\right)$$We analyze level $l$'s contribution in $O_1$ denoted by $O_{1,l,i}=\lpweight{l,i}|\levelval{l,i}-\levelval{l,j}|$. Honest nodes can input 1 to at most $\frac{\maxrange}{\overshoot{l}}$ checkpoints at level $l$, bounded by the total number of honest nodes \highthreshold. We denote the parameter $\checkpointcardi = \min(\frac{2\maxrange}{\overshoot{0}},\highthreshold{})$, where the weighted average can have a non-zero weight for at most \checkpointcardi{} checkpoints.
	We refer to these intervals using the terms \intweight{k,i} and \intmidpoint{k}, where $k\in\{1,\hdots,\highthreshold\}$. Additionally, $\lpweight{l,i}\leq\levelweight{l,i}$ based on~\cref{fig:multlevel}. We rewrite the above equation with these variables.
	\begin{gather*}
		O_{1,l,i} \leq \levelweight{l,i}\left|\frac{\sum_{k=1}^{\checkpointcardi}\intweight{k,i}\intmidpoint{k}}{\sum_{k=1}^{\highthreshold}\intweight{k,i}}-\frac{\sum_{k=1}^{\checkpointcardi}\intweight{k,j}\intmidpoint{k}}{\sum_{k=1}^{\highthreshold}\intweight{k,j}}\right|
	\end{gather*}
	We write $\intweight{k,j} = \intweight{k,i}+\adddiff{k}$.
	\begin{gather*}
		O_{1,l,i} \leq \levelweight{l,i}\left|\dfrac{\splitdfrac{(\sum_{k=1}^{\checkpointcardi}\intweight{k,i}\intmidpoint{k})(\sum_{k=1}^{\checkpointcardi}\intweight{k,i}+\adddiff{k})}{-(\sum_{k=1}^{\checkpointcardi}(\intweight{k,i}+\adddiff{i})\intmidpoint{k})(\sum_{k=1}^{\checkpointcardi}(\intweight{k,i}))}}{(\sum_{k=1}^{\checkpointcardi}\intweight{k,i})(\sum_{k=1}^{\checkpointcardi}\intweight{k,j})}\right|\\
        \implies O_{1,l,i} \leq \levelweight{l,i}\left|\frac{\sum_{k=1}^{\checkpointcardi}\adddiff{k}(\levelval{l,i}-\intmidpoint{k})}{\sum_{k=1}^{\checkpointcardi}\intweight{k,j}}\right|\\
        \implies O_{1,l,i} \leq \frac{\checkpointcardi\levelweight{l,i}\aaprecision\honestrange}{\sum_{k=1}^{\checkpointcardi}\intweight{k,j}}
	\end{gather*}
	As $\levelweight{l,i}=\max\{\intweight{k,i}: k\in\{1,\hdots,\checkpointcardi\}\}$, the term $\frac{\levelweight{l,i}}{\sum_{k=1}^{\checkpointcardi}\intweight{k,j}} \leq 1+\frac{\aaprecision}{\levelweight{l,j}} \leq 2$. Therefore, the term $O_1$ is upper bounded as follows: $
		O_1 \leq 2\maxrange\levelmax\aaprecision\checkpointcardi
    $.
	
	Now, we examine the second term $O_2 = \sum_{l=0}^{\levelmax}\adddiff{l}|\levelval{avg,i}-\levelval{l,j}|$. From~\cref{lem:level_weight_bound}, we know that the level weights $|\levelweight{l,i}-\levelweight{l,j}|\leq\aaprecision$, meaning $\adddiff{l}\leq \aaprecision\forall l\in\levelset$. Further, after a level $l > \log_2(\frac{\honestrange}{\overshoot{0}})+1$, we know $\levelweight{l,i}=0$ and $\adddiff{l}=0$. Finally, from the Validity of \protocol{} in \cref{thm:int_validity}, \levelweight{avg,i} always lies at a distance at most \honestrange{} from other honest inputs \oracleinput{i}.  Finally, the difference $|\levelval{avg,i}-\levelval{l,j}|<3\honestrange$. This gives  $O_2\leq 6\levelmax\maxrange\aaprecision$.
	Overall,
	\begin{gather*}
		O \leq  2\maxrange\levelmax\aaprecision(\checkpointcardi+3)\leq4\maxrange\levelmax\aaprecision\checkpointcardi
	\end{gather*}
	The parameters $\checkpointcardi= \min(\dfrac{2\maxrange}{\overshoot{0}},n), \levelmax = \log(\dfrac{\maxrange}{\overshoot{0}})$. 
	
	Therefore, \protocol{} achieves \approxconsensus{} when run for $\bigo{\dfrac{\maxrange\levelmax\checkpointcardi}{\epsilon}}$ rounds.
\end{proof}

\subsection{Complexity Analysis}\label{subsec:comm}
We analyze the communication complexity of \protocol{}. A level $l$ has at most $\frac{\honestrange}{\overshoot{l}}$ checkpoints \intervalmidpoint{\checkpointindex}{}{l} to which honest nodes can input 1. Based on the optimization we described in~\cref{subsec:optimize_comm}, level $l$ requires \bigo{n^2\min(\frac{\honestrange}{\overshoot{l}},n)} bits of communication per round. Overall, \protocol{} uses \bigo{n^2\min(\frac{\honestrange}{\overshoot{0}},n\levelmax)} bits of communication per round. Combined over \roundmax{} rounds, this complexity totals to \bigo{n^2\min(\frac{\honestrange}{\overshoot{0}},n\levelmax)\log(\frac{\maxrange\levelmax\min(\frac{\maxrange}{\overshoot{0}},n)}{\epsilon})} bits. 

\revisionchanges{\paragraph{Setting parameters for \name{}} \name{} requires setting \overshoot{0} and \maxrange as global system parameters. For the scope of this paper, we statically set $\overshoot{0} = \epsilon$, which ensures the minimum Validity relaxation is $\epsilon$. We set the parameter \maxrange using the assumption that honest nodes’ inputs come from a thin-tailed probability distribution. We theoretically analyze \name’s communication complexity (in bits) under different distributions. We discuss the communication complexity under different $\maxrange,\epsilon,$ and $\honestrange$ in~\cref{tab:comm_analysis}. Empirically, we collect and analyze data from the underlying applications and statistically infer the maximum possible difference between honest inputs. In the absence of such historical data from the application, \maxrange can be set based on domain knowledge. For example, in a use case where nodes want to agree on the price of Bitcoin, we can set \maxrange to be the maximum possible price observed so far - for example, $\maxrange=100000\$$. This way of setting \maxrange degrades the performance of \name by increasing its round complexity. }

\paragraph{Analysis under probability distributions}\label{subsec:prob_analysis} 
Many applications model real-world events using probability distributions. This analysis enables us to estimate the parameters in use and accordingly benchmark \protocol{}'s performance. We consider two types of distributions: (a) Thin-tailed distributions like Normal, Gamma, and Lognormal, and (b) Fat-tailed distributions like Pareto. We use the extreme value theory and order statistics to conduct this analysis. We calculate \maxrange{} based on a parameter $\statparam{}$, where we ensure that the range $\honestrange\leq\maxrange$ with probability $1-\negl(\statparam{})$. 

Thin-tailed distributions are used in error modeling in practice. For instance, Normal, Lognormal, and Gamma distributions have been used in sensor noise modeling and modeling sizes of insurance claims, respectively~\cite{boland2007statistical,embrechts2013modelling}. The extreme value theory suggests that the difference \honestrange{} between the maximum and minimum of $n$ randomly drawn samples from both Normal and Gamma distributions follows a Gumbel distribution, with CDF of $F(x) = e^{-e^{-x}}$~\cite[page 155]{embrechts2013modelling}, and the mean of the distribution of \honestrange{} grows as \bigo{\log(n)} for Normal and Gamma distributions. We calculate \maxrange{} using \statparam{} and $n$ in the CDF. Using this estimation, we observe that $\maxrange = \bigo{\statparam{}\log(n)}$. With this \maxrange{} value, \protocol{}'s communication complexity is \bigo{n^2\frac{\honestrange}{\epsilon}(\log(\frac{\honestrange}{\epsilon}\log{\frac{\honestrange}{\epsilon}})+\log(\statparam{}\log{n}))}.

In the case of honest samples coming from distributions like Pareto or Loggamma distributions with shape parameter $\alpha$, which have relatively fatter tails, the range of honest inputs follows a Frechet distribution, with CDF $F(x) = e^{-x^{-\alpha}}$~\cite[page 153]{embrechts2013modelling}. The mean of this distribution grows as \bigo{n^{\frac{1}{\alpha}}} In this case, we observe that $\maxrange = e^{\statparam{}}n^{\frac{1}{\alpha}}$, which gives \protocol{}'s communication to be 
\bigo{n^2\frac{\honestrange}{\epsilon}(\log(\frac{\honestrange}{\epsilon}\log{\frac{\honestrange}{\epsilon}})+\frac{\statparam{}}{\alpha}\log{n}))} bits. 
 
\begin{table}
	\renewcommand{\thefootnote}{\fnsymbol{footnote}}
	\def\footnoteCC{\footnotemark[2]}
	\def\footnoteValidity{\footnotemark[4]}
	\def\footnoteDelphiBID{\footnotemark[5]}
	\def\footnoteAppxAggr{\footnotemark[6]}
	\def\footnoteDora{\footnotemark[1]}
	\def\footnoteFin{\footnotemark[3]}
	\centering
	\caption[Related work comparison]{\small Communication and round complexity of \protocol{} under input conditions
	}%
	\label{tab:comm_analysis}
	\newcolumntype{s}{>{\centering\hsize=.25\hsize}X}
	\newcolumntype{Y}{>{\centering\arraybackslash}X}
	\newcolumntype{Z}{>{\centering\arraybackslash\hsize=1.8\hsize}X}
	\newcolumntype{A}{>{\centering\hsize=2\hsize}X}
	\newcolumntype{d}{>{\centering\hsize=.6\hsize}X}
	\newcolumntype{k}{>{\centering\hsize=1.4\hsize}X}
	\newcolumntype{M}{>{\centering\arraybackslash\hsize=1.8\hsize}X}
	\centering
	\begin{tabularx}{\linewidth}[!ht]{d d Z Y}
		\toprule
		\multicolumn{2}{c}{Condition}
		&
		Communication
		&
		Rounds
		\\
		\multicolumn{1}{c}{\maxrange}
		&
		\multicolumn{1}{c}{\honestrange}
		&
		&
		\\
		\midrule
		\multicolumn{1}{c}{\bigo{\epsilon}}
		&
		\multicolumn{1}{c}{\bigo{\epsilon}}
		&
		\bigo{n^2\log(\frac{\honestrange}{\epsilon})}
		&
		\bigo{\log(\frac{\honestrange}{\epsilon})}
		\\
		\multirow{2}{*}{\bigo{f(n)\epsilon}} 
		&
		\bigo{\epsilon}
		&
		\bigo{n^2(\log(\frac{n\maxrange}{\epsilon})+\log\log(f(n)))} 
		&
		\bigo{\log(\frac{n\maxrange}{\epsilon})+\log\log(f(n))} 
		\\
		&
		\bigo{\maxrange}
		&
		\bigo{n^3\log(f(n))(\log(\frac{n\maxrange}{\epsilon})+\log\log(f(n)))} 
		&
		\bigo{\log(\frac{n\maxrange}{\epsilon})+\log\log(f(n))} 
		\\
		\bottomrule
	\end{tabularx}
	\begin{flushleft}
		{\small $\maxrange$, $\epsilon$ are statically set parameters, as specified in~\cref{alg:n_2_BA}, whereas range $\honestrange\leq\maxrange$ is a runtime parameter that changes with varying inputs. $f(n)$ is any function growing faster than $n$.\ \maxrange{} is a constant measured based on the underlying data distribution, whereas $\epsilon$ is statically set at a system level. Under $\maxrange$ like $\maxrange = poly(n)\epsilon$, the $\log\log(n)$ term becomes $\log(polylog(n))$.}
	\end{flushleft}
\end{table}

\section{Application}
Oracle networks play an important role in the web3 ecosystem. They provide trusted information about real-world events to blockchains, which utilize it to execute logic on smart contracts. Oracle networks empower multiple applications in the DeFi space, like provisioning loans and trading futures, by providing reliable market data about stocks and cryptocurrencies~\cite{chainlink2019usecases}. Oracle nodes also play an important role in decentralized insurance applications like crop insurance, providing metrics about critical variables like forecasted rainfall and potential drought~\cite{chainlink2019usecases}. 
In these cases, the network must provide values representative of the real-world event, even in the presence of malicious nodes.

An oracle network must provide an attested value within the range of honest nodes' inputs to the blockchain, which verifies and includes it in the ledger. Chakka \etal~\cite{kate2023dora} defined this as the Distributed Oracle Agreement (DORA) problem, where nodes must agree on a value within the range of honest inputs, assisted by the external blockchain modeled as an SMR channel. A strawman solution to solving DORA is for each honest oracle to broadcast a signature on its input \oracleinput{i}, collect \highthreshold{} other signed values, submit this list of \highthreshold{} values to the SMR channel, wait for the first list in the SMR channel, and output the median of this list. The properties of SMR ensure that all honest oracles see the same first list and reach an agreement. However, this approach requires the SMR channel to verify \bigo{n} signatures and include a message of size \bigo{n\secparam} in the ledger, which makes it inefficient and expensive.

Prior approaches like Chainlink's reporting protocol~\cite{breidenbach2021chainlink} and Chakka \etal~\cite{kate2023dora} generate an aggregated BLS signature of size \bigo{n+\secparam} on a single value within the range of honest inputs\footnote{Message size can be reduced to \bigo{\secparam} with threshold BLS signatures, but require a threshold setup}. In Chainlink, nodes use partially synchronous convex BA to reach an agreement and then create a signature on the output~\cite{breidenbach2021chainlink}. Chakka \etal{} builds on the described strawman approach by adding a round to generate a succinct signature~\cite{kate2023dora}, and uses the SMR channel for agreement. However, both approaches are computationally expensive with \bigo{n^2} signature verifications per node per attested value (while achieving adaptive security). Refer to \cref{tab:relwork_oracle} for a detailed comparison between these approaches. 

We solve the DORA problem in a computationally efficient manner using \name. We introduce an additional step after reaching \approxconsensus, where nodes round off their input to the closest integer multiple of $\epsilon$. Next, they broadcast a signature on their rounded values, wait for \minthreshold{} signatures on a given value, and aggregate them to form a succinct signed message.

\paragraph{Intuition} After rounding, honest nodes' outputs must be within at most two adjacent $\epsilon$-checkpoints, which implies at least \minthreshold{} signatures must be broadcasted for at least one checkpoint. Further, no other value can receive $\geq\minthreshold$ signatures. Therefore, \name{} solves the DORA problem without using any additional computation than what is required to send the output to the SMR channel. The rounding operation incurs a further $\epsilon$ relaxation to the Validity.
\begin{table*}
	\footnotesize
	\renewcommand{\thefootnote}{\fnsymbol{footnote}}
	\def\footnoteCC{\footnotemark[2]}
	\def\footnoteValidity{\footnotemark[4]}
	\def\footnoteDelphiBID{\footnotemark[5]}
	\def\footnoteAppxAggr{\footnotemark[6]}
	\def\footnoteDora{\footnotemark[1]}
	\def\footnoteFin{\footnotemark[3]}
	\centering
	\caption[Related work comparison]{\small Comparison of relevant oracle reporting protocols
	}%
	\label{tab:relwork_oracle}
	\newcolumntype{s}{>{\centering\hsize=.25\hsize}X}
	\newcolumntype{Y}{>{\centering\arraybackslash}X}
	\newcolumntype{L}{>{\centering\hsize=0.7\hsize}X}
	\newcolumntype{Z}{>{\centering\hsize=1.3\hsize}X}
	\newcolumntype{d}{>{\centering\hsize=.5\hsize}X}
	\newcolumntype{k}{>{\centering\hsize=1.2\hsize}X}
	\newcolumntype{M}{>{\centering\arraybackslash\hsize=1.5\hsize}X}
	\centering
	\begin{tabularx}{\linewidth}[!ht]{X d M L d d Z Y}
		\toprule
		\multirow{2}{*}{\textbf{Protocol}}
		&
		\multirow{2}{*}{\textbf{Network}}
		&
		\textbf{Communication}
		&
		\textbf{Adaptively}
		&
		\multicolumn{2}{c}{\textbf{Computation}}
		&
		\multirow{2}{*}{\textbf{Rounds}}
		&
		\multirow{2}{*}{\textbf{Validity}}
		\\
		
		&
		
		&
		\textbf{Complexity (in bits)}
		&
		\textbf{Secure?}
		&
		\textbf{Sign}
		&
		\textbf{Verf}
		&
		
		&
		\\
		\midrule
		Chainlink~\cite{breidenbach2021chainlink}
		&
		p sync.
		&
		\bigo{ln^3 + \secparam{} n^3}
		&
		\xmark
		&
		\bigo{1}
		&
		\bigo{n}
		&
		4
		&
		$[m,M]$
		\\
		DORA~\cite{kate2023dora}\footnoteDora{}
		&
		async
		&
		\bigo{ln^2+\secparam{} n^2}
		&
		\xmark
		&
		\bigo{1}
		&
		\bigo{n}
		&
		3
		&
		$[m,M]$
		\\
		\midrule
		\textbf{\protocol}\footnoteDora{}
		&
		async
		&
		\bigo{ln^2\frac{\honestrange}{\epsilon}(\log(\frac{\honestrange}{\epsilon}\log{\frac{\honestrange}{\epsilon}})+\log(\statparam{}\log{n}))}
		&
		\checkmark
		&
		0
		&
		0
		&
		\bigo{\log(\frac{\honestrange}{\epsilon}\log{\frac{\honestrange}{\epsilon}})+\log(\statparam{}\log{n})}
		&
		$[m-\honestrange-\epsilon,M+\honestrange+\epsilon]$
		\\
		\bottomrule
	\end{tabularx}
	\begin{flushleft}
		{\small
			\honestset{} is the set of all honest inputs, $l$ is the size of each input, $m = \min(\honestset), M=\max(\honestset)$, $\honestrange = M-m$ is the range of honest inputs, and $\secparam{}$ is the cryptographic security parameter. $l< \secparam{}$ in practice.\ 
			\textbf{SMR Channel} In all protocols, oracles produce an attested output and submit it to the blockchain for consumption by smart contracts\cite[Section 5.4]{breidenbach2021chainlink}. This step is common for all protocols, so we do not count its cost in this table.\
			\textbf{Adaptive Security} Chainlink~\cite{breidenbach2021chainlink} and DORA~\cite{kate2023dora} can be made adaptively secure at the expense of $\bigo{n}$ factor increase in communication and computation complexities. Chainlink also incurs a $\bigo{n}$ increase in round complexity. \
			\footnoteDora{}\textbf{Agreement} These protocols can produce multiple outputs with the specified Validity condition. DORA~\cite{kate2023dora} can produce at least \bigo{n} possible outputs, whereas \protocol{} can produce at most two possible outputs. The external blockchain orders them and consumes the first output. \ 
		}
	\end{flushleft}
\end{table*}

\section{Evaluation}\label{sec:eval}
We evaluate \name in the context of two applications: (a) an oracle network reporting the price of a cryptocurrency, and (b) a group of surveillance drones agreeing on the location of an object. We analyze the input data in each application and accordingly configure \protocol{} to maximize efficiency while ensuring security. 

\subsection{Oracle Network}
In the first application, we consider an oracle network that reports the trading price of Bitcoin, a prominent cryptocurrency, in US Dollars. 
We configure the network to produce one price report every minute.  
Nodes in the Oracle network aim to agree on the current price of Bitcoin. 
Once a minute, each node measures the trading price of Bitcoin by querying one or a set of prominent cryptocurrency exchanges and computing the median of responses. 
Nodes then input this price, \oracleinput{i}, into an instance of \protocol{} protocol and report the output as Bitcoin's trading price for that minute. 
\begin{figure}
	\centering
	\includegraphics[width=0.9\linewidth]{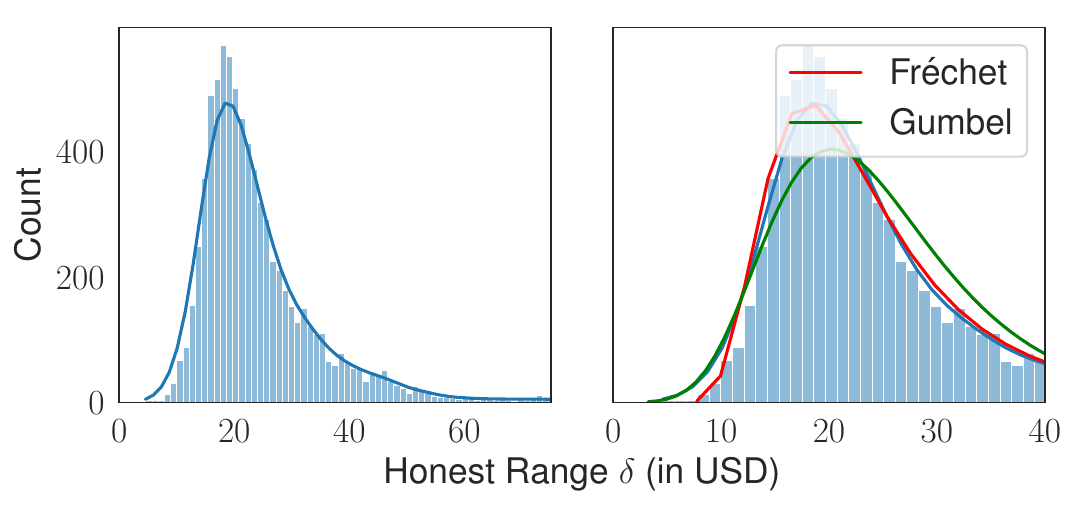}
	\caption{\small\textbf{Bitcoin price range histogram}: We plot a histogram of observed $\honestrange = |\max(\honestset)-\min(\honestset)|$ in US Dollars on the x-axis and the number of times the \honestrange values in the bin appeared in the two weeks on the y-axis. We also fit different probability distributions and observe that Fréchet and Gumbel distributions, the two distributions used to model extreme order values, are the closest fit, with Fréchet being the better fit. 
	}
	\label{fig:bitcoinpricehistogram}
\end{figure}

\paragraph{Range analysis} We configure \protocol{} based on a comprehensive analysis of the difference between honest inputs or range $\honestrange = \max(\honestset)-\min(\honestset)$, where \honestset is the set of inputs of honest nodes. We collected price feeds of Bitcoin from 10 prominent exchanges: \texttt{Binance}, \texttt{Coinbase}, \texttt{Crypto.com}, \texttt{Gate.io}, \texttt{Huobi}, \texttt{Mexc}, \texttt{Poloniex}, \texttt{Bybit}, \texttt{Kucoin}, and \texttt{Kraken}, for a two-week timespan from July 7th to July 21st, 2023, at a granularity of one reading every minute. We then calculated the maximum difference between the prices reported by these exchanges each minute, giving us one value per minute for two weeks. This value represents the difference between honest nodes' inputs or \honestrange because each node queries at least one exchange. We then plot a histogram of all these \honestrange values based on their frequency of occurrence in~\cref{fig:bitcoinpricehistogram}.

We observe that beyond $\honestrange=30\$$, higher $\honestrange$ values become increasingly improbable. The \honestrange values are below $100\$$ for $99.2\%$ of the time and below $300\$$ for $100\%$. We also fit various probability distributions for this data and observed the Frechet distribution with $\alpha=4.41$ and a scale factor of 29.3 to be the best fit for this data, implying the underlying distribution of inputs is a Loggamma distribution. We configure \protocol{}'s parameters using these probability distribution parameters. We derive maximum range $\maxrange=2000\$$ such that the range $\honestrange>\maxrange$ only with probability $p = 1-10^{-10}$, which guarantees a statistical security of $\statparam{}= 30$ bits. In this context and under these distribution parameters, assuming we run one instance of \protocol{} every minute, this assumption gets broken once every 2000 years. We set $\overshoot{0} = \epsilon = 2\$$.
\subsection{Distributed CPS}
In the second application, we consider a group of drones flying in an area, aiming to detect and localize objects. Such systems have gained increased prominence in light of recent advances in autonomous and semi-autonomous CPS~\cite{bandarupalli2021droneswarms,bandarupalli2023vega}. The drones are equipped with cameras to take photos of the area and are guided by the GPS navigation system. Each drone detects cars in the area using a deep learning-based object detection program like EfficientDet~\cite{tan2020efficientdet}. This program takes an image as input and outputs a rectangular bounding box of relative coordinates in the image containing the object. Each drone uses the center of this bounding box and its own geo-location provided by the GPS to estimate the geo-location of the target object. Given that the bounding box coordinates measured by drone $i$ is \locationcalc{BB,i}, its height is \droneheight{}m, and the drone's 2D location from GPS is \locationcalc{GPS,i}, the estimate of the target's 2D location \locationcalc{T,i} is calculated as 
$	\locationcalc{T,i} =  \locationcalc{BB,i}+\locationcalc{GPS,i}$.
The drones use this estimated location as an input to \protocol{} protocol to agree on the geo-location of each target in the area. As input $\locationcalc{T,i}=(x,y)$ is a 2D vector, drones use two instances of \protocol{} to agree on each coordinate,  $\locationcalcx{T,i},\locationcalcy{T,i}$, individually. 

\paragraph{Range Analysis} We analyze the coordinate-wise ranges $\honestrangex = |\max(\locationcalcx{T,i})-\min(\locationcalcx{T,j})|$ and $\honestrangey = |\max(\locationcalcy{T,i})-\min(\locationcalcx{T,j})|$ by first analyzing the distribution of the error distance $\distanceerror{i} = |\locationcalc{T,i}-\locationcalc{GT}|$. \distanceerror{i} quantifies the accuracy of the estimate \locationcalc{T,i} with the true location \locationcalc{GT} of the object. We then use extreme value theory on this distribution to estimate the ranges \honestrangex,\honestrangey. We consider cars to be the objects that need to be localized. We assume the drones are deployed at height $\droneheight{}=45$m or $150$ft. The vectors \locationcalc{BB,i} and \locationcalc{GPS,i} both contribute to the parameter \distanceerror{i}.
\begin{figure}
	\centering
	\includegraphics[width=0.9\linewidth]{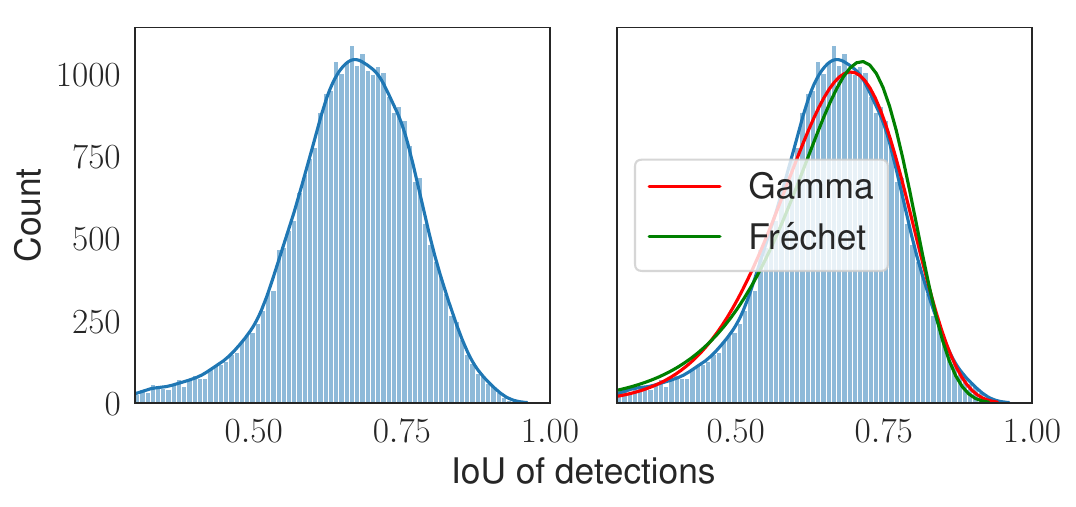}
	\caption{\small\textbf{IoU histogram for Drone-based object detection}: We plot the IoU values of detections output by the detection program and analyze their incidence by bucketing them into bins. These values follow a Gamma distribution, which has a thin tail. }
	\label{fig:iouhistogram}
\end{figure}

\paragraph{Intersection-over-Union(IoU) analysis} We empirically estimate the error induced by the object detector in \locationcalc{BB,i} by measuring and plotting the Intersection-over-Union(IoU) \iou{i} of the predictions. The IoU measures the overlap between the detector's output bounding box and the ground truth bounding box of the object. A higher IoU \iou{i} implies a more accurate detection. We characterize this IoU distribution by first training a model of EfficientDet~\cite{tan2020efficientdet} on a training dataset consisting of over 100000 cars~\cite{zhu2020vision,du2018unmanned,bandarupalli2023vega} and then testing it on a test dataset with 80000 cars. We then plot a histogram of all these IoU values based on their frequency of occurrence in~\cref{fig:iouhistogram}. We observe that EfficientDet produces detections with IoU following a Gamma distribution with mean $\iou{\mathsf{mean},i} = 0.87$. Further, EfficientDet outputs a detection with $\iou{i}<0.6$ IoU only in $0.37\%$ of the cases. 

We translate the IoU \iou{i} into the error distance vector $\distanceerror{i}$. The $x$ and $y$ coordinates of this vector are bounded by $\distanceerrorx{i},\distanceerrory{i}\leq(1-\iou{i})\times\diagonalbb$, where \diagonalbb is the length of the diagonal of the ground truth bounding box. We calculate \diagonalbb assuming by using the standard dimensions of $5$m length and $2$m width of a car. This formula gives us the distribution of the vector $\distanceerror{i} = (5.3\times(1-\iou{i}),5.3\times(1-\iou{i}))$, with a mean of $0.7$m.  

For \locationcalc{GPS}, the US FAA report on GPS accuracy documents the horizontal location error induced by GPS over 417 million samples~\cite{usnstb2021GPSRepost}. The accuracy is within $5$m $99.99\%$ percent of the time, with an average accuracy of $1.3$m.

\paragraph{Distribution fitting} Overall, the parameters \distanceerrorx{i},\distanceerrorx{i} have an expected value of $2$m and do not exceed $d=10.5$m $99.99\%$ of the time. However, analyzing the combined probability distribution is challenging because both \locationcalc{GPS} and \locationcalc{BB} are from different distributions. We perform asymptotic analysis by upper bounding the \locationcalc{GPS} distribution using a Gamma distribution of the same scale as \locationcalc{BB}. This approximation allows us to combine \locationcalc{GPS} and \locationcalc{BB} into a Gamma distribution with $scale=0.18$ and $shape=30.77$. For large $n$, extreme value theory suggests that the range of $n$ independently drawn inputs from a Gamma distribution is a Gumbel distribution~\cite[page 155]{embrechts2013modelling}. From the analysis in~\cref{subsec:prob_analysis}, a maximum range $\maxrange = \bigo{\statparam{}\honestrange_{\mathsf{mean}}}$. In this context, we set $\overshoot{0}=\epsilon=0.5$m and $\maxrange = 50$m. 
\begin{figure*}[t]
	\begin{multicols}{3}
		\begin{subfigure}[t]{\linewidth}
			\centering
			\includegraphics[width=\linewidth]{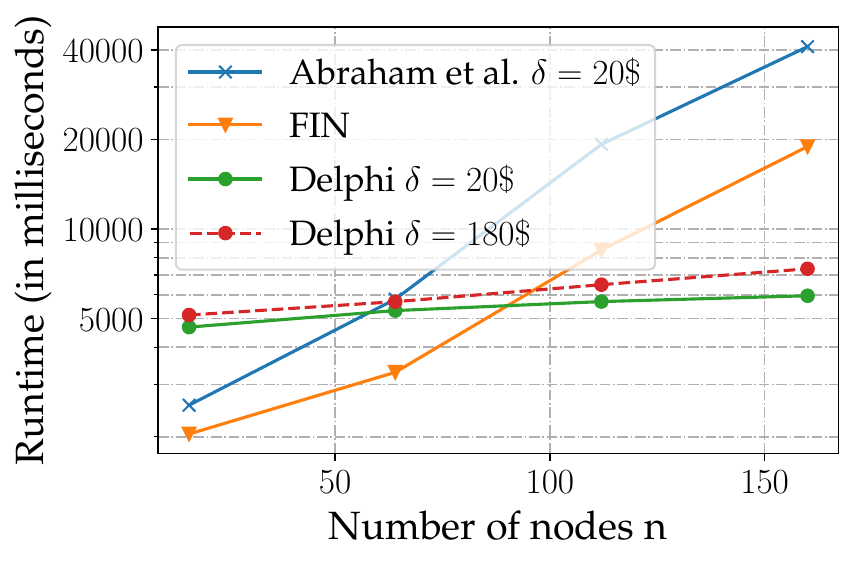}
			\caption{\small\textbf{Runtime vs $n$ on AWS}: We report the runtime of nodes reaching \approxconsensus on the price of Bitcoin. \honestrange is the range of honest inputs in US Dollars. \protocol{}'s config is $\overshoot{0}=10\$, \maxrange = 2000\$, \epsilon=2\$$. \protocol{}'s runtime increases much slowly with $n$ than FIN and Abraham \etal, and does not vary much with range $\honestrange$.}
			\label{fig:latencyaws}
		\end{subfigure}
		\begin{subfigure}[t]{\linewidth}
			\centering
			\includegraphics[width=\linewidth]{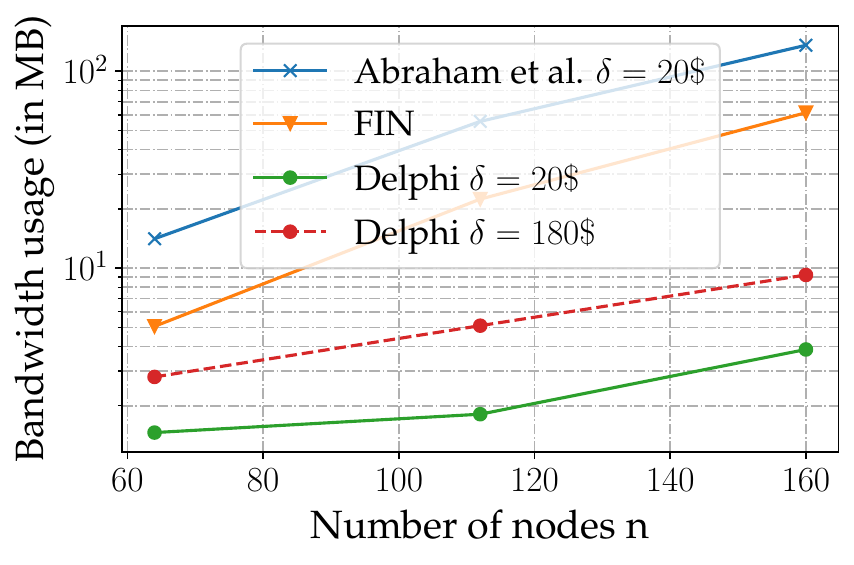}
			\caption{\small\textbf{Network bandwidth vs $n$ on AWS}: We report the amount of data used by each protocol to reach an agreement on the price of Bitcoin. \protocol{}'s config is $\overshoot{0}=\epsilon=2\$, \maxrange = 2000\$$. \protocol{}'s bandwidth usage is an order of magnitude lesser than FIN and Abraham \etal and also grows at a slower pace.}
			\label{fig:networkbandwidthaws}
		\end{subfigure}
		\begin{subfigure}[t]{\linewidth}
			\centering
			\includegraphics[width=\linewidth]{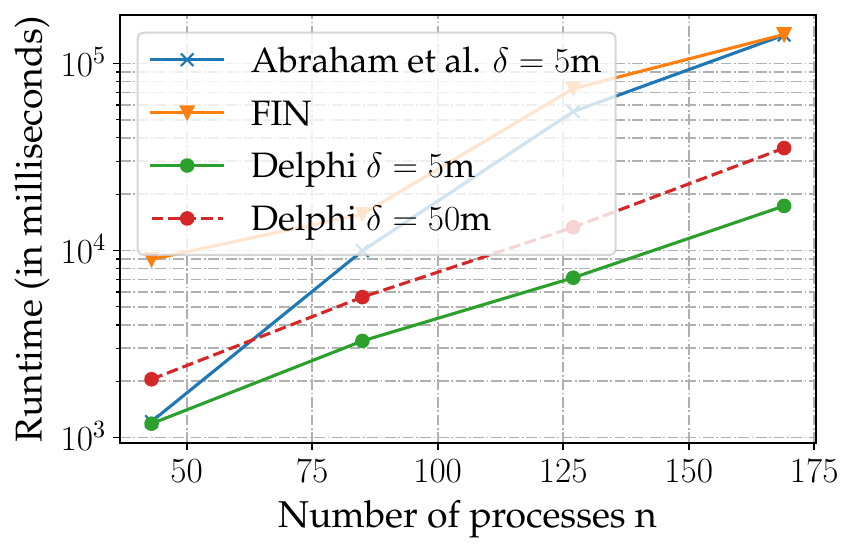}
			\caption{\small\textbf{Runtime vs $n$ on Embedded testbed}: We report the time taken to reach \approxconsensus on the location of a car seen by a group of drones. \honestrange is the range of honest inputs in meters. \protocol{}'s config is $\maxrange = 50m, \overshoot{0}=\epsilon=50cm$. \protocol{} terminates an order faster than FIN and Abraham \etal In contrast with AWS, a higher range \honestrange takes more time to terminate in this testbed.}
			\label{fig:latencyrpi}
		\end{subfigure}
	\end{multicols}
\vspace{-5mm}
	\caption{Scalability results of \protocol{}}
	\label{fig:delphi_comp}
\end{figure*}
\subsection{Implementation and Testbed Details}
We implement \protocol{} in Rust with the \texttt{tokio} library as our asynchronous runtime\cite{bandarupalli2024delphiRS} \footnote{Available at https://github.com/akhilsb/delphi-rs}. We use Hash-based Message Authentication Codes (HMAC) with the \texttt{SHA256} Hash function and shared symmetric keys to implement authenticated channels. We evaluate both protocols in two testbeds - (a) A geo-distributed testbed on AWS for the oracle network application and (b) A distributed testbed on Raspberry Pi devices for the object detection application. 

\paragraph{AWS testbed} We evaluate \protocol on Amazon Web Services (AWS) with varying nodes $n=16,64,112$, and $160$. We run both protocols on \texttt{t2.micro} nodes, each with 1 CPU core and 2GB RAM. We create a geo-distributed testbed of $n$ nodes to simulate execution over the internet. We distribute the nodes equally across 8 regions: N. Virginia, Ohio, N. California, Oregon, Canada, Ireland, Singapore, and Tokyo. 

\paragraph{Embedded Device Testbed} We set up an embedded testbed consisting of 15 Raspberry Pi 4-B devices connected with the help of a network switch functioning as a router. These devices are used in most commercially available Drones and other CPS. Each device has 2 GB of RAM and 4 CPU cores. We emulate a large network of drones by running multiple processes of \protocol{} on each device. We evaluate both protocols with varying nodes $n=43,85,127, 169$. 

\paragraph{Comparison to prior works} We compare \protocol{} to FIN~\cite{duan2023sigfreeacs}, the State-of-the-Art ACS protocol, and Abraham \etal~\cite{abraham2004approxagreement}, the current best approximate agreement protocol. We also implement FIN and Abraham \etal in Rust\cite{bandarupalli2024delphiRS}.  
\begin{figure}
	\centering
	\includegraphics[width=\linewidth]{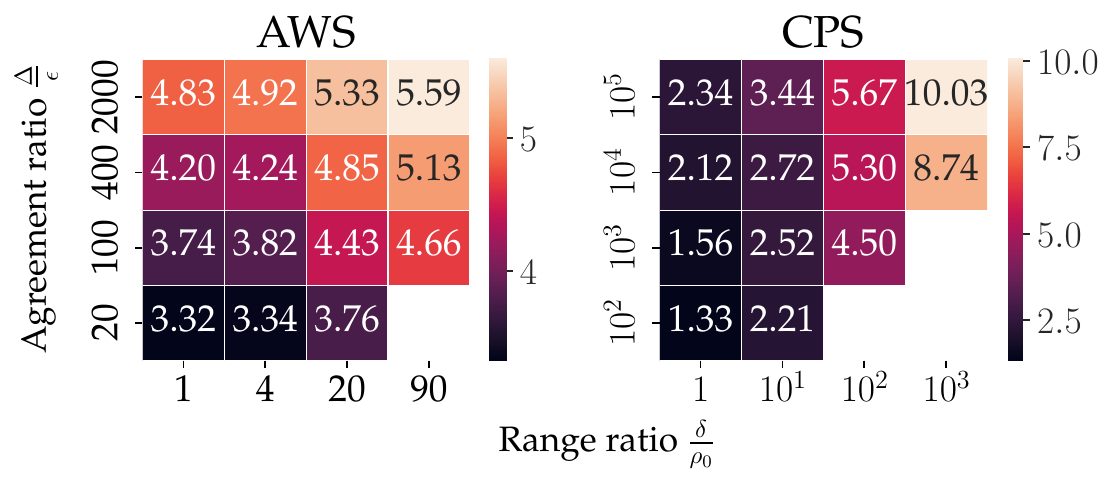}
	\caption{\small\textbf{Runtime patterns of \protocol{} on AWS and CPS}: We study the runtime (in seconds) of \protocol{} with varying agreement ratio $\frac{\maxrange}{\epsilon}$ which controls the round complexity, and range ratio $\frac{\honestrange}{\overshoot{0}}$, which controls the per-round communication. On AWS, the round complexity substantially influences the runtime compared to per-round communication. Whereas in CPS, we see the opposite pattern where the per-round volume of communication has more influence over runtime.}
	\label{fig:heatmapruntime}
\end{figure}

\subsection{Results} 
We run the three protocols - \protocol{}, Abraham \etal, and FIN - in the described applications, each in their respective testbed. We configure \protocol{} according to the calculated parameters for each application. We run the oracle network in the geo-distributed AWS testbed and the drone-based object detection in the CPS testbed. We measure the runtime and consumed network bandwidth of all three protocols. We evaluate \protocol{} under two \honestrange values: $\honestrange = 20\$,180\$$ in the oracle network, and $\honestrange = 5m, 50m$ in the CPS testbed to measure the average case and worst case runtime of \protocol{}. We plot our findings in~\cref{fig:delphi_comp}. 

\paragraph{Scalability Results} We plot and compare the runtimes of the protocols in the oracle network in~\cref{fig:latencyaws}. \protocol{} takes only $\frac{1}{3}$rd and $\frac{1}{6}$th the time taken by FIN and Abraham \etal at $n=160$, which demonstrates its scalability. The reason for \protocol{}'s higher runtime at small $n$ is the high round complexity of \protocol{} coupled with the round trip time of a geo-distributed network. However, \protocol{}'s runtime scales much better mainly because of its computational efficiency and practical communication complexity. Even at a higher $\honestrange$ value, \protocol{} considerably outperforms both prior works because of its low message complexity. We also plot the network bandwidth consumed by each protocol in~\cref{fig:networkbandwidthaws}. \protocol{} consumes an order of magnitude lower bandwidth in both cases of \honestrange than FIN. 

We also plot the runtime of the protocols in the drone-based object detection application in~\cref{fig:latencyrpi}. The devices in this testbed have lower computing power than the AWS testbed, increasing the weight of a protocol's computational efficiency. At $n=169$ processes, \protocol{}'s runtime is $\frac{1}{8}$th that of FIN and Abraham \etal in the average case and the worst case of $\honestrange=5,50$m, respectively. Unlike in the AWS testbed, where the round complexity played an important role at low $n$, \protocol{} outperforms FIN in this application at all $n$ because of their computational efficiency. We also observe that \protocol{}'s runtime is impacted significantly by the range \honestrange, which is in contrast with \protocol{}'s runtime in AWS. 

\paragraph{Runtime analysis of \protocol{}} We investigate the major factors affecting \protocol{}'s runtime in both testbeds. We consider two parameters - (a) The agreement ratio $\frac{\maxrange}{\epsilon}$ and (b) The range ratio $\frac{\honestrange}{\overshoot{0}}$, which control the round complexity and per-round communication complexity. In~\cref{fig:heatmapruntime}, we plot a heatmap of the runtime of \protocol{} with $n=64$ and $n=85$ processes in the AWS and CPS testbeds. We observe that the round complexity is the dominant driver of runtime in AWS. However, we observe the opposite in CPS, where per-round communication volume has a greater influence. The limited network bandwidth of the devices in CPS is the rate-limiting factor. Hence, the parameters $\epsilon$ and $\overshoot{0}$ are crucial for \protocol{}'s performance in AWS and CPS, respectively. 

\subsection{Validity Relaxation}
In the oracle network application, \protocol{}'s output is \honestrange{} distance away from the range of honest inputs. As we observe in~\cref{fig:bitcoinpricehistogram}, the mean \honestrange{} is $25\$$. Taking the current price of Bitcoin as 40000\$, \protocol{} adds $0.05\%$ error to the price value in expectation than FIN and Abraham \etal \revisionchanges{Concretely, \name{}'s output is 25\$ away from the average of honest inputs in expectation. For FIN and Abraham \etal, it is 12.5\$ away from the average in expectation. } Further, this error is less than 0.5\% for more than 99.2\% of the time. In the Drone surveillance application, the 
mean difference along each coordinate is $\honestrange_{\mathsf{mean}}\leq ln(n)\times0.18 = 0.92$m. Essentially, the output of \protocol{} is at most $1.3$m farther away from the target object than FIN and Abraham \etal in expectation. \revisionchanges{Concretely, \name{}'s output is $2.6$m away from the average of honest inputs in expectation. Whereas for FIN and Abraham \etal, the output is $1.3$m away from the average of honest inputs in expectation. }
\section{Related work}
Many agreement protocols were proposed with the weak Validity condition~\cite{mostefaoui2015signature,crain2020binaryba,abraham2022efficient}, which states that if all honest nodes have the same input $v$, then the output must be $v$. Binary BA protocols where nodes input only 0 or 1 offer this form of Validity. Prominent asynchronous binary BA protocols~\cite{mostefaoui2015signature,crain2020binaryba,abraham2022efficient} have an \bigo{n^2} communication complexity (even sub-quadratic~\cite{blum2020asynchronous}) and require \bigo{1} common coins. However, such protocols are unsuitable for applications with multi-valent inputs, where honest nodes often input multiple values. 

\textit{Validated} BA protocols have been proposed, which have an external validity property, ensuring that the output satisfies an external predicate. 
Validated BA (VABA)~\cite{abraham2019asymptotically} and Multi-valued Validated BA (MVBA)~\cite{cachin2001secure,lu2020dumbo,duan2023sigfreeacs} protocols have been proposed with external Validity with a communication complexity of \bigo{\secparam{} n^2}, and multi-valent common coins used for electing nodes as leaders. 
However, local Validity assessment cannot be done for real-world data like price feeds of a stock/cryptocurrency or the location of an object in an area.

Convex Validity enables nodes to agree on a value within the range of honest inputs. One way to achieve Convex Validity is to compute the Asynchronous Common Subset (ACS)~\cite{benor1994acs}, which enables agreement on a set of \highthreshold{} inputs. The median of the ACS is guaranteed to be within the range of honest inputs. Current ACS protocols can be classified as either being BKR-style~\cite{benor1994acs} or MVBA-style~\cite{cachin2001secure}. BKR-style ACS such as HoneyBadgerBFT~\cite{miller2016honey}, Beat~\cite{duan2018beat}, and PACE~\cite{zhang2022pace} use RBC and $n$ binary BA instances. However, these protocols are computationally expensive, requiring \bigo{n} common coins to terminate. A concurrent work  HashRand~\cite{bandarupalli2023hashrand} overcomes this bottleneck by building common coins from computationally efficient Hash functions. However, the described ACS protocols still cost a high \bigo{\secparam{} n^3} bits. 

The MVBA-style ACS uses RBC with MVBA to overcome the high runtime of BKR-style ACS.\@ In exchange for running $n$ binary BAs, MVBA uses a more complex multi-valent common coin for proposal election (PE), implemented using \bigo{\log(n)} binary coins~\cite{cachin2001secure}. Prior MVBA-style protocols like Dumbo2~\cite{guo2020dumbo} and SpeedingDumbo~\cite{guo2022speedingdumbo} trade computational complexity for communication efficiency and therefore use \bigo{n^2} signatures per ACS.\@ FIN-ACS~\cite{duan2023sigfreeacs} is the most computationally efficient ACS protocol thus far, which uses RBC and a multi-valent coin for PE.\@ FIN-ACS has an overall computational complexity of \bigo{\log(n)} coins and uses \bigo{\secparam{} n^3} bits, which we argue is also very expensive in our setting. Information-theoretic approaches like WaterBear~\cite{zhang2022waterbear} do not use threshold signatures but require \bigo{exp(n)} communication to terminate. 

\approxconsensus{} (AA) protocols have been proposed as an alternative to the compute-intensive exact agreement protocols. Dolev \etal~\cite{dolev1986reaching} proposed the first asynchronous AA protocol with $n=5t+1$ resilience, with \bigo{n^2} communication per round and \bigo{\log_2(\frac{\honestrange}{\epsilon})} rounds. Abraham \etal's~\cite{abraham2004approxagreement} protocol improved this resilience to $n=3t+1$ but requires \bigo{n^3} communication per round. Followup works explored AA over a partially connected network~\cite{vaidya2012iterative}, multi-dimensional AA~\cite{vaidya2013byzantine,mendes2015multidimensional}, and AA with asynchronous fallback~\cite{ghinea2022syncappxaggr,ghinea2023multidimappxaggr}. However, Abraham \etal's per-round communication of \bigo{n^3} bits has not been improved so far.
Binary AA is closely related to the notion of \emph{proxcensus} \cite{fitzi2021new,ghinea2022round}, a generalization of graded consensus.

\section{Conclusion}
The paper introduced Delphi, a deterministic asynchronous convex BA protocol designed for applications like distributed oracles in blockchains and fault-tolerant cyber-physical systems. Addressing the challenge of maintaining convex validity in sensor/oracle nodes measuring the same source, Delphi minimizes communication overhead to $\mathcal{\tilde{O}}(n^2)$ while keeping computation costs low. Unlike existing protocols relying on randomization or approximate agreement techniques with high computational or communication costs, Delphi combines agreement primitives and a novel weighted averaging technique. Experimental results demonstrate Delphi's superior performance, achieving an 8x and 3x improvement in latency over state-of-the-art protocols in CPS and AWS environments for a system with $n=160$ nodes.     
\section{Acknowledgments}
We thank the reviewers and our shepherd Darya Melnyk for the helpful feedback and suggestions to improve this draft significantly. We thank Manoj Patil and Saaransh Jakhar at Supra for collecting data on cryptocurrency price feeds. We also thank Azam Ikram and Mihir Patil for help with experiments on the CPS testbed. 
This work was supported in part by NIFA award number 2021-67021-34252, the National Science Foundation (NSF) under grant CNS1846316 and National Science Foundation Cyber Physical Systems (CPS). C. Liu-Zhang's research was supported by Hasler Foundation Project \#23090 and ETH Zurich Leading House RPG-072023-19.


\iffull{}
\fi

\ifsubmission{}
  \bibliographystyle{plain}
\else
  \bibliographystyle{plain}
\fi
\bibliography{references}


\end{document}
